\documentclass[11pt]{article}
\usepackage{amsmath}
\usepackage{graphicx,psfrag,epsf}
\usepackage{enumerate}
\usepackage{natbib}
\usepackage{url} 

\newcommand{\blind}{0}

\usepackage{soul}
\usepackage{algorithm}
\usepackage{algorithmic}
\usepackage{bbm} 
\usepackage{bm}
\usepackage{booktabs} 
\usepackage{setspace}
\usepackage{sectsty} 
\usepackage{float} 
\subsectionfont{\normalfont\bfseries\itshape}
\subsubsectionfont{\normalfont\itshape} 
\usepackage{mathptmx} 
\usepackage{multirow}
\usepackage[table]{xcolor}
\usepackage[toc,page]{appendix}
\usepackage{caption}
\usepackage{subcaption} 
\usepackage{listings}

\RequirePackage{amsthm,amsfonts,amssymb}
\RequirePackage[colorlinks,citecolor=blue,urlcolor=blue]{hyperref}
\newcommand{\appendixpagenumbering}{
  \break
  \pagenumbering{arabic}
}

\theoremstyle{plain}

\newtheorem{theorem}{Theorem} 
\newtheorem*{remark}{Remark} 
\newtheorem{lemma}[theorem]{Lemma}
\newtheorem{definition}{Definition}
\theoremstyle{remark}
\newtheorem*{example}{Example}

\usepackage{bbm} 
\usepackage{algorithm}
\usepackage{algorithmic}
\usepackage{booktabs} 
\usepackage{multirow}
\usepackage[table]{xcolor} 

\usepackage{subfiles}

\begin{document}

\def\spacingset#1{\renewcommand{\baselinestretch}%
{#1}\small\normalsize} \spacingset{1}


\makeatletter
\renewcommand\@date{{%
  \vspace{-\baselineskip}%
  \large\centering
  \begin{tabular}{@{}c@{}}
    Xinyu Zhang\textsuperscript{1} \\
    \normalsize xinyu$\_$zhang@ncsu.edu
  \end{tabular}%
  \quad and\quad
  \begin{tabular}{@{}c@{}}
    Sujit Ghosh\textsuperscript{1} \\
    \normalsize sujit.ghosh@ncsu.edu
  \end{tabular}

  \bigskip

  \textsuperscript{1}Department of Statistics, \\ NC State University

  \bigskip

}}
\makeatother

\if0\blind
{
  \title{\bf PaEBack: Pareto-Efficient Backsubsampling for Time Series Data}
  \maketitle

} \fi

\if1\blind
{
  \bigskip
  \bigskip
  \bigskip
  \begin{center}
    {\LARGE\bf PaEBack: Pareto-Efficient Backsubsampling for Time Series Data}
\end{center}
  \medskip
} \fi



\begin{abstract}
Time series forecasting has been a quintessential topic in data science, but traditionally, forecasting models have relied on extensive historical data. In this paper, we address a practical question: How much recent historical data is required to attain a targeted percentage of statistical prediction efficiency compared to the full time series? We propose the Pareto-Efficient Backsubsampling (PaEBack) method to estimate the percentage of the most recent data needed to achieve the desired level of prediction accuracy. We provide a theoretical justification based on asymptotic prediction theory for the AutoRegressive (AR) models.  In particular, through several numerical illustrations, we show the application of the PaEBack for some recently developed machine learning forecasting methods even when the models might be misspecified. The main conclusion is that only a fraction of the most recent historical data provides near-optimal or even better relative predictive accuracy for a broad class of forecasting methods.

\end{abstract}

\noindent%


\newpage
\spacingset{1.5} 

\section{Introduction}
\label{sec:intro}

Time series analysis has been a very matured area of research in a wide variety of scientific and social applications. For example, in finance and econometric research, securities and stocks often call for short and long term forecast values. In medical fields and healthcare research, electrocardiogram (ECG) forecasting is often of interest, and more recent years, it was almost essential to obtain forecast values of COVID-19 counts. Forecasting methods within time series analysis have progressed enormously in the past several decades using classical (stationary) time series models and (often non-stationary) Machine Learning (ML) models. We present a very brief overview of some of these models and associated methods of inference before we introduce our proposed approach to best utilize the historical data.

Traditional time series analysis can be traced back to the development of the exponential smoothing state space method (ETS) proposed by \cite{winters1960forecasting} as well as the celebrated autoregressive integrated moving average (ARIMA) models made available in the book by \cite{box1970time} over the past century along with several extensions. For instance, by explicitly modeling the conditional variance along with conditional mean, there are Autoregressive Conditional Heteroscedastic (ARCH) methods (\cite{engle1982autoregressive}), which have been extended to Generalized Autoregressive Conditional Heteroscedastic (GARCH)  methods (\cite{bollerslev1986generalized}), and variants like Glosten-Jagannathan-Runkle GARCH methods (gjrGARCH) (\cite{glosten1993relation}). Providing a comprehensive literature review of such developments is almost impossible, so we refer to some excellent existing reviews (\cite{chatfield2000time, armstrong2001principles}). In passing, we mention only a few relatively recent methods that we would feature as a part of the application of our proposed methods. Hence, our literature review is limited only to a small subset of articles that are relevant to our models featured later in our article.

Based on traditional time series models, various effective time series forecasting methods have been developed to address different issues in the recent decade. \cite{de2011forecasting} modeled complex seasonality through trigonometric representations, namely TBATS, while \cite{scott2014predicting} developed the Bayesian structural time series (BSTS) method to filter spurious time series features. Another method called Prophet was proposed by \cite{taylor2018forecasting} to incorporate customizations on trends, seasonality, and holidays. With modern advances in machine learning, time series forecasting has benefited from nonlinear tools such as autoregressive neural networks (ARNNs) (\cite{faraway1998time}), ensemble deep learning methods (\cite{ray2021optimized}), and other hybrid or ensemble methods that combine the advantages of classical models and advanced nonlinear machine learning techniques, albeit at the cost of higher computational power.

Forecasting the number of cases of Coronavirus poses challenges due to extreme uncertainty and non-stationarity. \cite{chakraborty2022nowcasting} have presented a rather comprehensive review of various statistical and machine learning methods to address these issues, ranging from ARIMA, ARNN, to Long-Short Term Memory (LSTM)(\cite{ceylan2020estimation}, \cite{chimmula2020time}).
For stock price forecasting, investigations ranging from Machine Learning techniques such as Artificial Neural Networks (ANNs), Support Vector Machine (SVM), random forests, and naive-Bayes, to Deep Learning frameworks such as Convolutional Neural Networks (CNNs), Deep Belief Networks (DBNs), and LSTM (\cite{patel2015predicting}, \cite{sezer2020financial} ). The problem that remains a bit elusive for almost all of these forecasting methods is how much of the historical data is required.


In the era of big data, where data collection has become more accessible, obtaining high volumes of data has become easier. However, when the primary goal is short-term forecasting, using a high-resolution long series of historical data may not be necessary. Most time series data exhibits a feature of diminishing autocorrelation, and long time-series data rarely abides by strict assumptions like strong and/or weak stationarity. Hence, a longer training data series may not significantly improve relative predictive accuracy in practice (see Section \ref{sec:met} for more details).
\cite{smyl2016data} have reported that the most recent time steps contribute the most to forecasting immediate future points. A natural question is to what extent such high volumes of historical data are useful for predicting relatively short series of future values. To explore this, consider a series of auto-correlated data of size $n$ denoted by $X_{1:n}=\{X_1, \ldots,X_n\}$ and the goal is to predict the next $h$ values represented by $X_{(n+1):(n+h)} = \{X_{n+1},\ldots,X_{n+h}\}$. We aim to determine how many of the {\em most recent} $k$ values, denoted by $X_{(n-k+1):n}=\{X_{n-k+1},\ldots,X_n\}$, are practically sufficient to provide accurate forecasts of $X_{(n+1):(n+h)}$ relative to using the full data $X_{1:n}$.

We illustrate the concept with a simple example in Figure \ref{fig:ora} using a simulated AR(5) process.
\begin{figure}[h] 
	\centering
	\includegraphics[width=\textwidth]{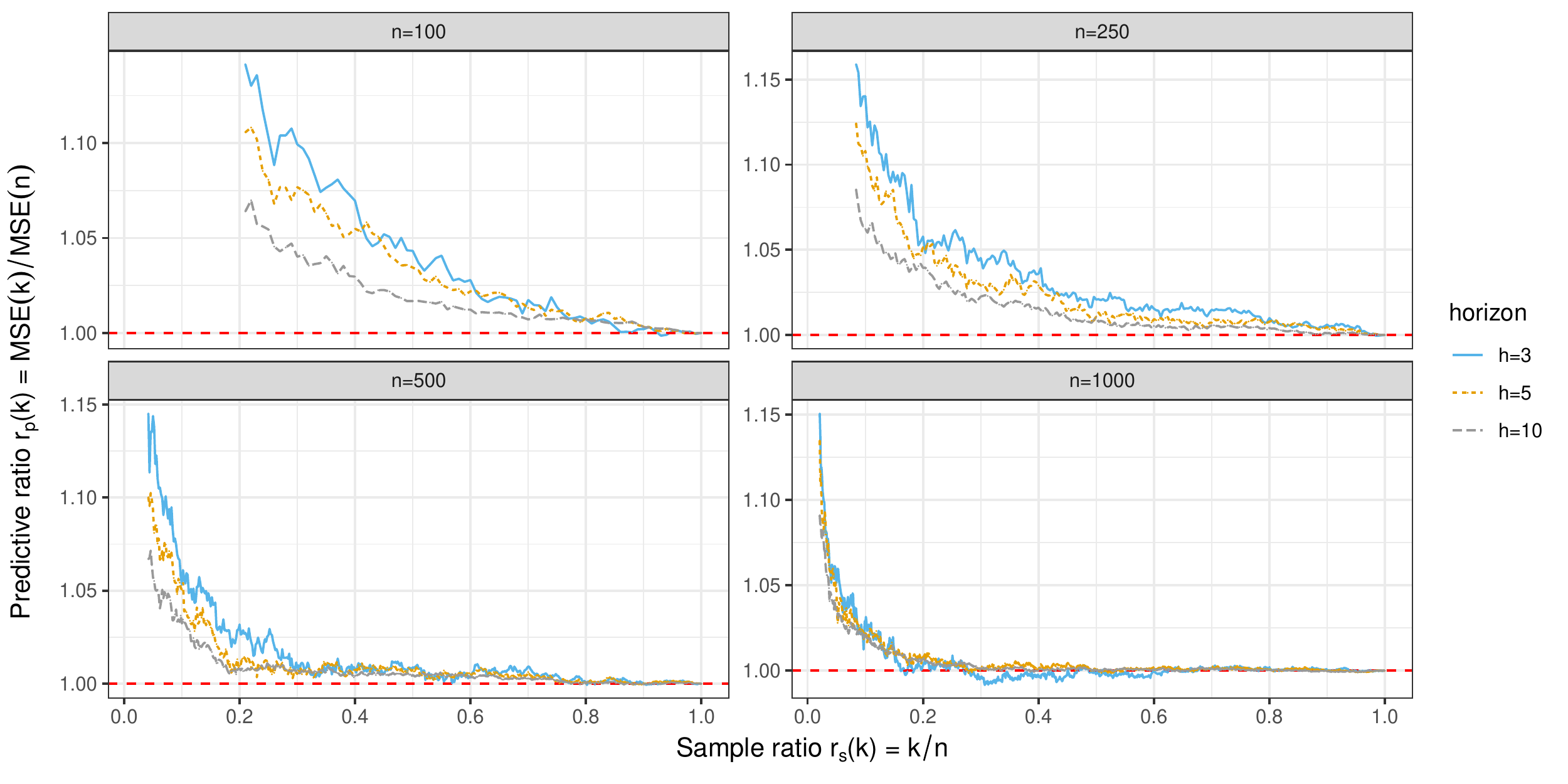}
	\caption{A motivating example of short-term forecasting of stationary data using the most recent $k$ observations. The data is generated from a stationary AR(5) process with autoregressive coefficients $\phi=(0.5, -0.4, 0.3, -0.2, 0.1)^T$, and different sample sizes $n=100, 250, 500, 1000$ are considered. The AR(5) model is fitted using $X_{(n-k+1):n}$ to predict $h$ future values $X_{(n+1):(n+h)}$, where $h=3$ (blue solid), $h=5$ (orange dotted), and $h=10$ (grey dashed). 
The horizontal axis of the plot represents the sample ratio $r_s(k)= \frac{k}{n}$, the proportion of recent observations used for forecasting, with a lower value indicating better sample efficiency. 
The vertical axis illustrates the relative predictive efficiency denoted by $r_p(k)=\frac{MSE (k)}{MSE (n)}$ on the test set, with lower values indicating better efficiency.
The experiment is repeated 1000 times to obtain the Monte Carlo median of $r_p(k)$ for varying $r_s(k)$. As a baseline, a horizontal red dotted line at $r_p(k)=1$ is marked in the plots, representing the predictive ratio when using the full sample size $n$ for forecasting.}
	\label{fig:ora}
\end{figure}
As one would expect, it is clearly evident from Figure \ref{fig:ora} that the predictive error $MSE(k)$ of using the most recent $k$ observations decreases as $k$ increases. Nevertheless, notice that by using only about $20$\% (as $r_s(k)\approx 0.2)$ most recent observations (say when $n=500$ or $1000$), the prediction efficiency (denoted by $r_p(k)$, defined in the caption of the Figure \ref{fig:ora}) achieved by the most recent $k$ values is almost as good as that achieved by the full sample. The goal then becomes estimating the smallest $k$ for which $r_p(k)\geq 1-\epsilon$ with high probability, given any small $\epsilon>0$.

In certain domains, such as finance or medical devices, short-term prediction with a limited number of steps ahead (e.g., $h=10$) can be of primary interest and sufficient for making policy or health status decisions. 
For example, in finance, daily stock prices older than a year or two may not provide much additional information, and ten days ahead forecasts might be sufficient for decision-making. Similarly, in ECG analysis, one-day or one-week ahead forecasts can provide informative insights for users to take appropriate actions (\cite{degiannakis2018forecasting, fan2019forecasting}) with medical devices.
Thus, extracting valuable information solely from a small subset of the most recent data can significantly reduce the burden of large-scale data storage and processing in an era marked by widespread digitization.
Additionally, the assumption of long time series being strictly or weakly stationary is often unrealistic in practice. As a result, researchers have developed local stationary (LS) models (\cite{dette2020prediction}) that assume a slowly changing characteristic of the stochastic process (\cite{nason2000wavelet}). These models may consider the series as an autoregressive process with locally varying parameters over time (we refer to \cite{zhao2015inference, roueff2016prediction, kley2019predictive} for illuminating examples and theory).
In this study, we aim to estimate the optimal subsample size $k$ using a proposed schematic approach called \underline{Pa}reto-\underline{E}fficient \underline{Back}subsampling ({\bf PaEBack}, {\em pronounced pay-back}) method for time series data.

The selection of optimal subsample size $k$ relies on both the employed forecasting model and the chosen discrepancy criteria for evaluating forecasting accuracy. The autoregressive (AR) approximation has been extensively validated as a practical reliable approximation for time series data with theoretical guarantees of accuracy (e.g., \cite{goldenshluger2001nonasymptotic, kley2019predictive}).

In this paper, we present the general framework of PaEBack in Section \ref{sec:met}, and introduce the concept of Pareto optimal efficiency. For clarity and ease of explanation, we introduced the theoretical guarantee under AR setting in Section \ref{sec:ar}; however, the core essence of this paper lies in applying the PaEBack framework to a wide range of time series models, which need not be strictly stationary, such as nonlinear AR process-based or local stationary process-based time series models. 
In Section \ref{sec:sim}, we demonstrate the versatility of the PaEBack framework by employing various time series forecasting models and discrepancy measures under different assumptions. Additionally, we present practical applications of the PaEBack framework using advanced machine learning techniques in the time series analysis of financial and epidemiological data in Section \ref{sec:dat}. Finally, in Section \ref{sec:con}, we summarize concluding remarks and discussions.

\section{PaEBack Under General Setting } \label{sec:met}

Pareto efficiency, also known as Pareto optimality or Pareto superiority, is a concept from economics representing a state in which no individual can be made better off without making someone else worse off. In other words, an allocation of resources or a situation is considered Pareto efficient if it is not possible to make any improvements that benefit one party without negatively impacting another party.

In the context of time series models, Pareto efficiency can be related to the trade-off between the amount of past data used for modeling and the statistical efficiency of the model. Suppose that only a portion of past data is sufficient to achieve a certain level of statistical accuracy or predictive power in a time series model. This means that using additional historical data beyond this point might not significantly improve the model's performance and could potentially introduce noise or unnecessary complexity.
In this scenario, Pareto efficiency would imply that the chosen subset of past data provides the best balance between model performance and computational efficiency. Including more data points may lead to diminishing returns in terms of model improvement while increasing computational inefficiency. Therefore, the model is Pareto efficient if no further improvements can be achieved in terms of predictive accuracy or other relevant metrics by including more past data points.
Pareto efficiency in the context of time series models reflects the idea that there is an optimal point at which the benefits of using additional data are outweighed by the costs, and further optimization would involve a trade-off between the interests of different stakeholders, such as model accuracy, computational resources, and simplicity.

\begin{figure}[h]
	\begin{center}
		\includegraphics[scale=0.5]{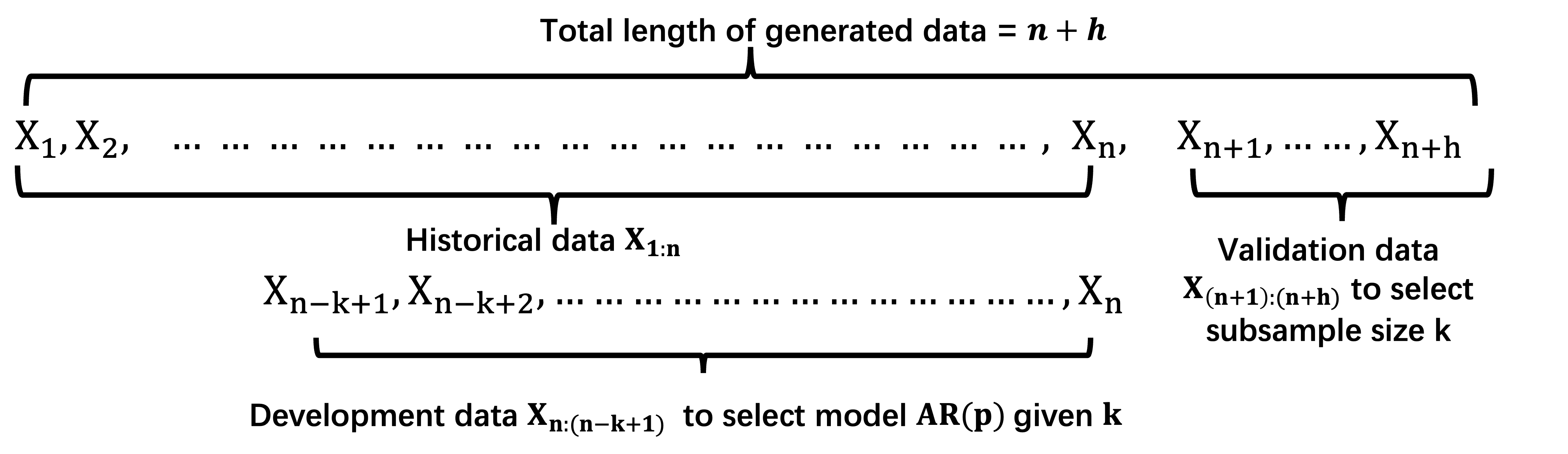}
	\end{center}
	\caption{Data splitting strategy for PaEBack. The data splitting strategy employed in the PaEBack method involves the following steps. Firstly, for each different fixed development size $ k$'s, the predictive model is trained using the development set, denoted as $X_{(n-k+1):n}$, which includes the most recent $k$ data points from the historical data. Secondly, once the predictive model is trained, it is used to generate forecasts for the validation set, denoted as $X_{(n+1):(n+h)}$. The validation set represents a sequence of data points following the development set and extends $h$ steps into the future. This validation set serves as the benchmark for evaluating the predictive performance of using different values of the development size $k$.}
	\label{fig:bs}
\end{figure}

\subsection{General Notations}
Consider ${X}_{1:t} \stackrel{\Delta}{=} \{X_1, X_2, \ldots, X_t\}$, representing the set of time series values observed up to time $t$ (discretized in units such as minutes, days, or weeks). Our primary objective is to select a subsample $X_{(n-k+1):n}$ of size $k$ within the development set, as depicted in Figure \ref{fig:bs}, to forecast the future $h$ steps, $X_{(n+1):(n+h)}$ in the validation set, with both predictive accuracy and computational efficiency in mind. This study aims to estimate the optimal subsample size $k$ using the proposed schematic approach called PaEBack for Time Series data, as illustrated in Figure \ref{fig:bs}.

For each time point $X_{n+t}$, $t=1, \cdots, h,$ in the validation set, denote its conditional expectation given the past $k$ observations in the development set as $X^{t}_{n, k} = \mathbb{E}\left(X_{n+t} \mid X_{(n-k+1):n} \right)$, then its approximation error is $e_{\text{app} }= X_{n+t}-X_{n,k}^t$, dominated by model assumptions. Denote the predicted value of $X_{n+t}$ as $\hat{X}^{t}_{n, k}$, then the corresponding estimation error depending on the model fitting procedure is $e_{\text {est}} = X_{n,k}^t - \hat{X}^{t}_{n, k}$, where its variability increases as the model dimension increases and generally decreases with increasing sample sizes under some stationarity assumptions.

Define the forecasting error at time point $n+t$ for $(t\in\{1, \cdots, h\})$ based on $X_{(n-k+1):n}^n$ as $e _{n,k}^t$, then it can decomposed as:
\begin{align*}
e_{n,k}^t  & = X_{n+t}- \hat{X}^{t}_{n, k} \\
& = \left( X_{n+t}- X_{n,k}^t \right) +  \left (X_{n,k}^t - \hat{X} _{n,k}^t  \right) \\
& = e_{\text {approx}} + e_{\text {est}} 
\end{align*}
Then the accumulative average predictive error on the validation set $X_{ (n+1) : (n+h)}$ when using development subsample $X_{(n-k+1):n}$ can be quantified using the mean squared error $MSE(k)$, where
\begin{equation} \label{eq:mse}
MSE(k) = \mathbb{E} \left( \frac{1}{h} \sum_{t=1}^h { \left(e_{n,k}^t \right) }^2 \right).
\end{equation}

Achieving a high practical prediction accuracy (such as MSE based on validation data) and a high computational efficiency often involves a trade-off. 
When dealing with weakly stationary series, the prediction accuracy is expected to increase with an increasing $k$ if more distant past samples are utilized. However, this would require the assumption of stationarity to be valid for a long extended series. For example, utilizing only 20-25 minutes of the most recent electrocardiogram (ECG) data can yield an accuracy of over $90\%$, compared to using the complete history of ECG, which may span several hours or days. This demonstrates that the percentage of accuracy gained from utilizing the entire historical data may not be practically beneficial. Similarly, in financial stock price prediction, relying on only the most recent daily data for a few months suffice to accurately forecast prices one week ahead instead of utilizing years of daily data.
In light of these considerations, we present the PaEBack method in this paper, which addresses the optimal selection of subsample size $k$ to strike a balance between prediction accuracy and computational efficiency.

\subsection{Pareto Optimal Efficiency for Time Series Data}
\label{sec:bs}

In this section, we strive to achieve a balance between practical prediction accuracy and computational efficiency by adopting the concept of Pareto optimal efficiency. Pareto optimality, also known as Pareto efficiency, arises when multiple efficiency goals are pursued, and no adjustment can be made to improve one objective without adversely affecting another. In our approach, sample efficiency $r_s(k)=\frac{k}{n}$ is the ratio of the chosen sample size to the complete dataset. Predictive efficiency $r_p(k) = \frac{MSE(k)}{MSE(n)}$ is the ratio of prediction accuracy achieved using the past $k$ sample $X_{(n-k+1):k}$ to that obtained from utilizing the full data $X_{1:n}$, with MSE as our exemplary metric. We define the \textbf{Pareto optimal efficiency} in Definition \ref{def:dual}.

\begin{definition}\label{def:dual}
Define the Pareto optimal efficiency as the combination of the sample efficiency quantified by the ratio $r_s(k)$ and predictive efficiency quantified by the ratio $r_p(k)$ respectively:
\begin{equation}
\begin{array}{rl}
\text{\textbf{sample efficiency: }} & r_s(k) = \frac{k}{n} ,  \\
\text{\textbf{predictive efficiency: }} & r_p(k) = \frac{MSE(k)}{MSE(n)},
\end{array}
\end{equation}
where $MSE (k)$ represents the expected Mean Squared Error of forecasting values in the validation set $X_{(n+1):(n+h)}$ predicted by the model using training data from the development set $X_{(n-k+1):n}$. 
For the purpose of discussion, we employ MSE as an illustrative discrepancy criterion to measure forecasting accuracy and Yule-Walker estimation as an example to illustrate computational efficiency. It is worth noting that various estimation methods and alternative criteria can also be employed, and these are discussed further in this section.
\end{definition}

Using fewer training samples can avoid excessive data storage and maintenance and make it less likely to violate assumptions such as local stationarity rather than stationarity. Consequently, a smaller sample efficiency ratio $r_s(k)$ indicates the achievement of a higher computational efficiency. On the other hand, a smaller value of $r_p(k)$ indicates better relative predictive performance of the forecasting model when using the reduced data set $X_{(n-k+1):n}$.

In seeking dual Pareto optimal efficiency, we favor smaller values of $(r_s(k), r_p(k))$ as we vary $k=1,2, \ldots$. These values correspond to the most favorable trade-offs between computational efficiency and predictive accuracy for different subsample sizes, thereby enabling us to identify the most efficient and accurate forecasting model.

\begin{remark}[Generalizability]
Aside from MSE, various discrepancy criteria can also be used to measure the predictive accuracy of the test set, such as Mean Absolute Error (MAE), mean absolute percentage error (MAPE), root-mean-square error (RMSE), and Symmetric mean absolute percentage error (SMAPE) (see \cite{hanh2018lasso, efron2004least, syntetos2005accuracy}). 
Recall the notation of true value at time point $n+t$ is $X_{n+t}$ and its corresponding prediction as $\hat{X}_n^t(k)$, then the criteria mentioned above applied to the test set $X_{(n+1):(n+h)}$ can be expressed as follows:
\begin{eqnarray}
    \arraycolsep=5.4pt \def\arraystretch{0.2}
 \begin{array}{ll}
    \operatorname{MAE} =\frac{\sum_{t=1}^h\left|X_{n+t}-\hat{X}_n^t(k)\right|}{h},  &  \operatorname{MAPE}=\frac{1}{h} \sum_{t=1}^h \left|\frac{X_{n+t}-\hat{X}_n^t(k)}{ X_{n+t}}\right|, \\
    \mathrm{RMSE} =\sqrt{\frac{\sum_{t=1}^h\left( X_{n+t} - \hat{X}_n^t(k) \right)^2}{n}}, & \operatorname{SMAPE} =\frac{100 \%}{h} \sum_{t=1}^h \frac{\left|X_{n+t}-\hat{X}_n^t(k)\right|}{\left(\left|X_{n+t}\right| + \left|\hat{X}_n^t(k)\right| \right) / 2}.
    \end{array}
\end{eqnarray}
Similarly, the use of the $\mathcal{O}(n)$ Yuler-Walker method can be replaced by other methods of estimation such as simple moving average (SMA), exponential smoothing, and aggregation methods (see \cite{syntetos2005accuracy, de2011forecasting, kourentzes2014improving}) depending on the assumed properties of the time series models. 
\end{remark}

\begin{definition} \label{def:pi}
For time series data consisting of $n$ observations, define $\epsilon_n$ as the practical irrelevancy, a PaEBack framework is said to be $(1-\epsilon_n)$ dual efficient if there exists $k=k_n(<n)$ such that (i) $\liminf_{n\rightarrow\infty}\Pr(r_p(k_n)<1+\epsilon_n)=1$ and (ii) $\limsup_{n\rightarrow\infty} {k_n\over n}<1$ when $\epsilon_n\rightarrow 0$ as $n\rightarrow\infty$.
\end{definition}

We utilize the standard large-sample theory for stationary time series (e.g., see \cite{wu2011asymptotic}) to provide theoretical justifications for the proposed PaEBack framework. For stationary time series, it is expected that the predictive ratio $r_p(k)\geq 1$. However, in a finite sample, the practical predictive efficiency $r_p(k)$ is a random quantity with potentially complex expressions, making it challenging to control. To address this issue, we aim to find an asymptotically equivalent value denoted by $Ar_p(k)$ as $n\rightarrow\infty$. This allows us to understand how fast the practical irrelevancy $\epsilon_n$ can approach zero while we accept a small and practically negligible efficiency loss $\epsilon_n>0$.

To derive a simplified expression for the asymptotic predictive efficiency $Ar_p(k)$ such that $r_p(k)/Ar_p(k)\stackrel{p}{\rightarrow} 1$ as $n\rightarrow\infty$, we derive an asymptotic result within AR model framework. In the ensuing theorem, we derive the optimal PaEBack development sample size, denoted by $k^{opt}$, that satisfies $r_p(k)\leq 1+\epsilon_n$, where $\epsilon_n$ is allowed to approach zero at a certain rate. This provides a theoretical insight into the trade-off between predictive efficiency and the sample size in the PaEBack framework. Extension of such asymptotic theoretical derivations for more general models is left as a part of future work.

\section{PaEBack Under AR Setting} \label{sec:ar}
We assume that the discrete time series ${X_1, X_2,\ldots}$ follows a stationary $AR(p)$ process of order $p$ and can be represented by the following equation:
\begin{equation} \label{eq:ar}
X_{t} ={\phi}_{1} X_{t-1} +{\phi}_{2} X_{t-2} +\cdots +{\phi}_{p} X_{t-p}  + \epsilon_{t}, \quad \forall t \geq p+1,
\end{equation}
where $\phi=(\phi_1,\phi_2,\ldots,\phi_p)^T$ denotes the vector of AR coefficients that satisfies the weak stationarity condition, with all roots of the equation $1-\phi_1x - \phi_2x^2 -\cdots - \phi_p x^p = 0$ lying outside the unit circle. The sequence of errors $\{\epsilon_{p+1}, \epsilon_{p+2},\ldots\}$ is assumed to satisfy a white noise process with zero mean and constant variance $\sigma^2$.

For each time point $X_{n+t}$, $t=1, \cdots, h,$ in the validation set, denote
its conditional expectation under AR setting is then
\vspace{-0.3cm}
\begin{align} \label{eq:ywh}
X^{t}_{n, k} = \mathbb{E}\left(X_{n+t} \mid X_{(n-k+1):n} \right) = \sum_{i=1}^{ \lfloor t-1, p \rfloor } \phi_i X^{t-i}_{n,k} + \mathbbm{I}(t \leq p) \sum_{i=t}^p\phi_i X_{n+t-i},
\end{align}
with approximation error as
\begin{align*} 
e_{\text{app} }= X_{n+t}-X_{n,k}^t  = \sum_{i=1}^{ \lfloor t-1, p \rfloor } \phi_i\left(  X_{n+t-i} -  X_{n, k}^{t-i} \right) + \epsilon_{n+t}.
\end{align*}

Denote by ${\phi}_{i,k}$ as an estimator of the AR coefficient $\phi_i$, $i=1, \cdots, p$ (as defined in (\ref{eq:ar})) based on past $k$ observations $X_{(n-k+1):n}$ (e.g., one choose to use the popular Yule-Walker (YW) methods), the predicted value of $X_{n+t}$ is then expressed as:
\begin{align}
\hat{X}^{t}_{n, k} = \sum_{i=1}^{ \lfloor t-1, p \rfloor } {\phi} _{i,k} \hat{X} _{n,k}^{t-i} + \mathbbm{I}(t \leq p) \sum_{i=t}^p {\phi} _{i,k} X_{n+t-i}.
\end{align}
Then the corresponding estimation error depending on the model fitting procedure, is given by:

\begin{align} \label{eq:est}
e_{\text {est}} &= X_{n,k}^t - \hat{X}^{t}_{n, k}   = \sum_{i=1}^{ \lfloor t-1, p \rfloor } \left( \phi_i X_{n,k}^{t-i} -  {\phi} _{i,k} \hat{X} _{n,k}^{t-i} \right) + \mathbbm{I}(t \leq p) \sum_{i=t}^p
\left( \phi_i - {\phi} _{i,k} \right) X_{n+t-i},
\end{align}
where its variability increases as the model order $p$ increases and generally decreases with increasing sample sizes under some stationarity assumptions. The following result provides an asymptotic expression for the optimal subsampling size under a set of regularity conditions.

\begin{theorem}[] \label{thm:1} Denote $A=\sum_{j=1}^h \sigma_j^2$ and $B=\sigma^2 \sum_{j=1}^h tr\{M_j^{T} \Gamma^{-1} M_j \Gamma \}$, where the variance $\sigma^2_j$ and matrix $M_j$ are defined in Lemma \ref{prop:a1h} and Lemma \ref{prop:sigma2h} in Appendix \ref{appA} respectively:
\begin{itemize}
    \item[(i)] The asymptotic value of the predictive efficiency $r_p(k)$ is given by
	\vspace{-0.2cm}
	\begin{align} \label{eq:aasr}
	r_p(k) \sim Ar_p(k) & = 1 + \frac{\frac{n}{k}-1}{1 + n \frac{A}{B}} \xrightarrow{n\rightarrow\infty} 1+\frac{1}{k} \cdot \frac{B}{A}  \quad \mbox{as}\; n\rightarrow\infty. 
	\end{align}
	In other words, $r_p(k)/Ar_p(k)\stackrel{p}{\rightarrow} 1$ as $n\rightarrow \infty$.

    \item [(ii)] Consider a sequence $\epsilon_n \rightarrow 0$ as in Definition \ref{def:pi}, such that $n\epsilon_n \rightarrow \lambda>0$. The Pareto optimal efficient subsampling size $k^{opt}$ that satisfies $Ar_p(k) \leq 1+\epsilon_n$ is given by
	\vspace{-.2cm}
	\begin{align} \label{eq:askopt}
	k^{opt} \sim n \left(  \frac{1}{1+\lambda \frac{A}{B}}  \right).
	\end{align}
\end{itemize}	
\end{theorem}

\noindent The proof of the above result is provided in Appendix \ref{appA}. Interestingly, the asymptotic ratio $\frac{A}{B}$ does not depend on $\sigma^2$. Thus, the asymptotically optimal PaEBack subsampling size $k$ depends on the AR models only through a function of AR coefficients $\phi_j$s.

\begin{example}
For numerical illustration, Figure \ref{fig:ora} provides the dual efficiency for a simple simulated scenario by fitting an oracle AR process with true order known using the Yule-Walker (YW) estimations. The horizontal axis represents the sample efficiency $r_s(k)$ while the vertical axis represents the predictive efficiency $r_p(k)$. As $r_p(k)$ is smaller than one in reality, the actual predictive model using fewer samples $X_{(n-k+1):n}$ achieves higher forecasting accuracy without including all the historical data. Hence, we expect the pair $(r_s(k), r_p(k))$ to get closer to the left bottom corner to indicate better Pareto efficiency. Under this ideal scenario, the predictive efficiency decreases as the development sample size increases at a certain parametric rate. However, when $h=3$ and $n=1000$, we could observe the blue predictive efficiency curve drops even below the horizontal line at $r_p(k)=1$ when the $r_s(k)$ is around $0.3 \sim 0.6$, which shows one can achieve almost the same (or even better) predictive accuracy by utilizing only $30\% \sim 60\%$ of the immediate past observations. 

Taking this setting to demonstrate Theorem \ref{thm:1}, we first notice that the true value of the asymptotic ratio $\frac{{A}}{{B}} = 0.1525$ (see Section \ref{appA2} for further numerical details). Comparing the asymptotically optimal value of $k^{opt}$, we could take $\lambda$ ranging from 4.37 to 15.30. Thus, the estimated bound of the error of the asymptotic predictive ratio $\epsilon_n$ (defined in Theorem \ref{thm:1}) is then $[0.004, 0.015]$, which amounts to an efficiency loss of only $0.4$\% to $1.5$\% relative to the sample efficiency gain of 40-70\%. We also provide a numerical illustration for a real data example by estimating the ratio $A/B$ in Section \ref{appA3} for a series of stock price data. Notice that the ratio $A/B$ is a continuous function of the AR coefficients; it can be consistently estimated using the YW estimates of the AR coefficients.
\end{example}

\subsection{PaEBack with Order Selection}
\label{sec:ordersel}
When fitting it to time series data, the assumption of knowing the true order $p$ of an autoregressive (AR) model is often unrealistic. To address this issue, we relax this assumption and examine the PaEBack dual efficiency achieved when the true order of the AR process is unknown and selected using penalized methods.

Various penalized least square methods, such as the Least Absolute Shrinkage and Selection Operator (LASSO) and adaptive LASSO, have been widely employed for time series order selection (e.g., \cite{zou2006adaptive, wang2007regression, nardi2011autoregressive, hanh2018lasso}). LASSO, originating from linear models, is a well-known variable selection method that minimizes the squared loss with an $\mathit{l}_1$ penalty on the regression coefficients. The LARS algorithm (\cite{efron2004least}) is often used to realize LASSO. To account for the varying importance of different parameters, the adaptive LASSO estimator (also originating from multiple linear regression models) incorporates lag information in time series data (\cite{zou2006adaptive}). However, it is known to have limitations when the predictors exhibit high collinearity, which is often the case with time series data when lagged values are chosen as predictor variables. To address this issue, \cite{zou2009adaptive} extended the adaptive LASSO to the adaptive elastic-net method by introducing a quadratic regularization term.

In Section \ref{sec:221}, we introduced the Pareto Efficient Backsubsampling method for Time Series data via sliding window (PaEBack-SW) approach for parameter tuning. In section \ref{sec:osm}, we delve into the order selection methods for the autoregressive time series modeling. Given the commonly observed sparsity among different lags in autoregressive time series and the decay of correlations between distant time points as lags increase, we present an algorithm to leverage this information by customizing the adaptive weights in Section \ref{sec:aw}.

\subsubsection{PaEBack Parameter Tuning via Sliding Window (PaEBack-SW)} \label{sec:221}

Traditional cross-validation methods, such as 5-fold CV, are commonly employed for parameter tuning; however, with time series data, it is essential to preserve the natural time order during data splitting. Therefore, a sliding window approach is adopted to select the tuning parameter for penalized least squares methods in time series data.
 
Let $p_m$ be the upper bound for the order of AR models. The PaEBack-SW approach divides the development data $X_{(n-k+1):n}$ into $k_m=k-p_m$ pairs of training and testing sets respectively, as illustrated in Figure \ref{fig:sw}. In this approach, for each row $i$, the training set $X_{(n-k+i):(n-k+p_m-1+i)}=\{X_{n-k+i}, \cdots, X_{n-k+p_m-1+i}\}$ is used to fit an AR model with order no greater than $p_m$. The performance of the model is then evaluated using $X_{n-k+p_m+i}$ as a test case, assessing its predictive capability.

\begin{figure}[h] 
	\centering
	\includegraphics[width=.7\textwidth]{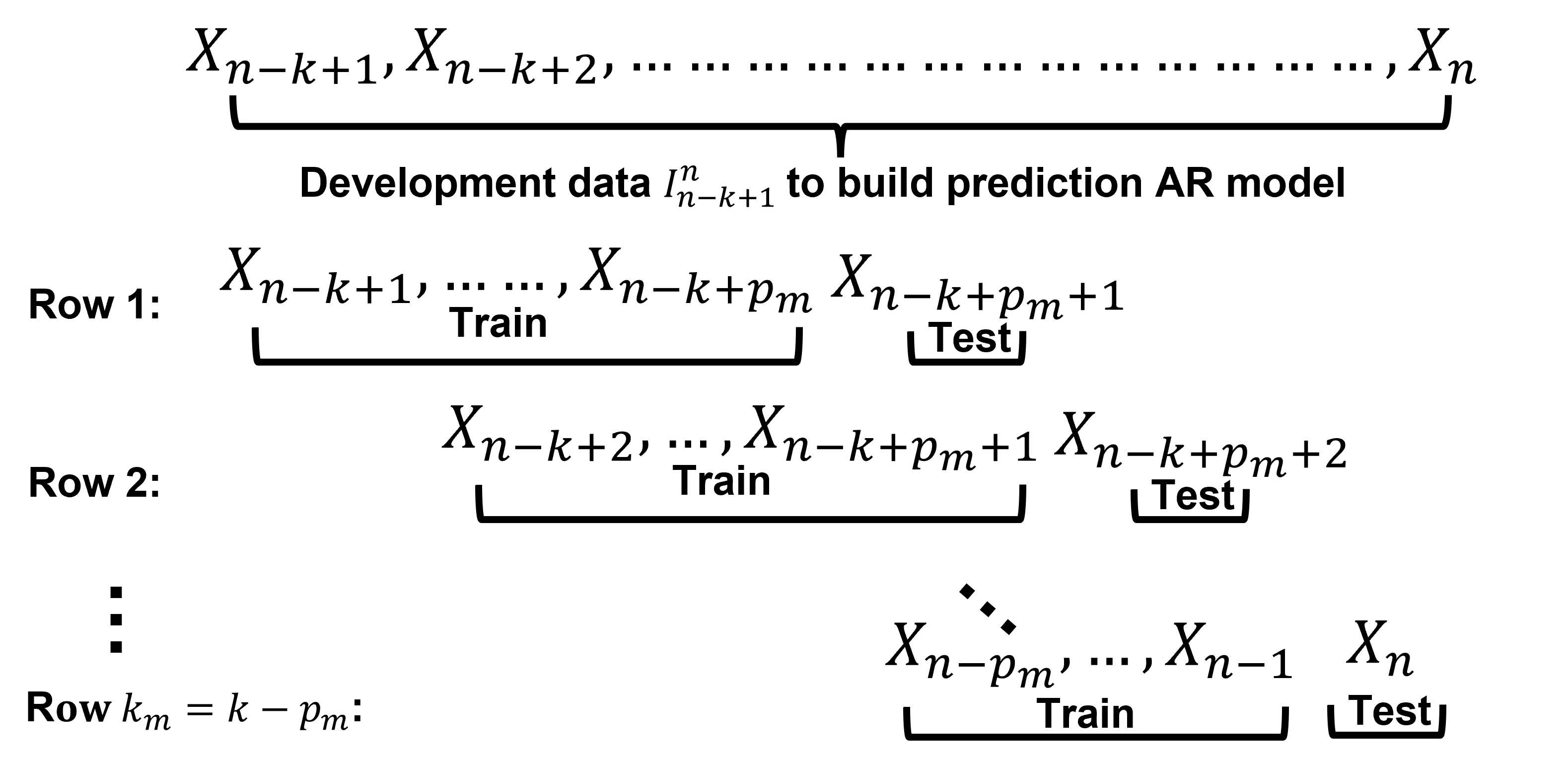}
	\caption{PaEBack-SW, model selection procedure of the PaEBack method using sliding window approach, where the development sample size $k$ is assumed to be at least $p_m$.}
	\label{fig:sw}
\end{figure}
Following the ideas of autoregressive regression (\cite{nardi2011autoregressive, hanh2018lasso}), we can express the time series forecasting problem as a linear model:
\vspace{-0.3cm}
\begin{equation} \label{eq:lm}
\mathbf{y} = \mathbf{Z} \mathbf{\phi} + \mathbf{\epsilon}
\end{equation}
\vspace{-.3cm}
where the response variable $\mathbf{y}_{(k_m : 1)}$, design matrix $\mathbf{Z}_{(k_m : p_m)}$, coefficient vector $\mathbf{\phi}_{(p_m : 1)}$, and error term $\mathbf{\epsilon}_{(k_m : 1)}$ are defined as:
\begin{align} \label{mod:1}
&\mathbf{y}_{(k_m : 1)} \stackrel{\Delta}{=} \left( \begin{array}{c}
X_{n-k+p_m+1} \\ X_{n-k+p_m+2} \\ \vdots \\ X_n
\end{array} \right), \quad
\mathbf{Z}_{(k_m: p_m)} \stackrel{\Delta}{=}\left( \begin{array}{cccc}
X_{n-k+p_m} & X_{n-k+p_m-1} & \cdots & X_{n-k+2} \\
X_{n-k+p_m+1} & X_{n-k+p_m} & \cdots & X_{n-k+2} \\
\vdots & \vdots & & \vdots \\
X_{n-1} & X_{n-2} & \cdots & X_{n-p_m} \\
\end{array} \right) \\
&\mathbf{\phi}_{(p_m: 1)} = (\phi_1, \phi_2, \cdots, \phi_{p_m})^T, \quad
\mathbf{\epsilon}_{(k_m: 1)} = (\epsilon_{n-k+p_m+1}, \cdots, \epsilon_{n})^T. \nonumber
\end{align}
The linear model formulation (\ref{eq:lm}) matches the model selection procedure illustrated in Figure \ref{fig:sw} with, each row of the design matrix $\mathbf{Z}_{(k_m : p_m)}$ and the response variable $\mathbf{y}_{(k_m : 1)}$ in (\ref{mod:1}), matching to a specific pair of the training and testing set in the figure; while each element of $\mathbf{\epsilon}_{(k_m : 1)}$ in (\ref{mod:1}) is the associated white noise as defined in Eq. (\ref{eq:ar}).

\subsubsection{Autoregressive Order Selection using Penalized Methods} \label{sec:osm}
We obtain the penalized estimate of the AR coefficient vector by extending the penalized least squares method known as adaptive LASSO (ALASSO), proposed by \cite{zou2006adaptive} to time series data, which solves the following optimization problem:
\vspace{-.5cm}
\begin{align}
\label{eq:adalasso}
\widehat{\boldsymbol{\phi}}(\text { AL })=\underset{\phi \in \mathbbm{R}^p}{\arg \min }\|\mathbf{y}-\mathbf{Z} \boldsymbol{\phi}\|_{2}^{2}+\lambda \sum_{j=1}^{p} \hat{w}_{j}\left|\phi{j}\right|,
\end{align}
with the adaptive weights:
\vspace{-.1cm}
\begin{equation}
\hat{w}_{j}=\left(\left|\hat{\phi}_{j}^{i n i}\right|\right)^{-\gamma}, \quad j = 1, \cdots, p,
\label{eq:weights}
\end{equation}
where $\gamma$ is a positive constant set to 1 and $\widehat{\boldsymbol{\phi}}^{i n i}$ is an initial root-$k$ consistent estimate of $\phi$, such as the ordinal least square estimation when $p_m < k_m $ or the YW estimate with subsample size $k$. However, while ALASSO achieves sparsity by using the $\mathcal{L}_1$ penalty, it is well known that it suffers from biases due to high multicollinearity, which occurs for time series data as the columns of the design matrix $\mathbf{Z}$ formed by the lagged values of the time series (see Eq. (\ref{mod:1})). Consequently, \cite{zou2009adaptive} broadened the ALASSO by adding a ridge-type $L_2$-penalty to the coefficients, and the adaptive elastic-net estimator is given by solving the following optimization problem:
\begin{equation} \label{eq:ae}
\widehat{\boldsymbol{\phi}}(\text {AE})=
\left( 1 + \frac{\lambda }{2k}(1 - \alpha) \right)
\left\{
\underset{\phi \in \mathbbm{R}^p}{\arg \min }\|\mathbf{y}-\mathbf{Z} \boldsymbol{\phi}\|_{2}^{2}
+ {  \frac{\lambda (1-\alpha)}{2}  }  \|\boldsymbol{\phi}\|_{2}^{2}
+ { \frac{\lambda \alpha}{2}} \sum_{j=1}^{p} \hat{w}_{j}\left|\phi{j}\right|
\right\}.
\end{equation}
Here we use the sliding window data split method to select tuning parameters, $\lambda$ and $\alpha$, in the adaptive elastic-net estimation. The weights of $\hat{w}_j$ given can penalize different parameters adaptively, with the asymptotic theory of ALASSO to establish the oracle consistency under linear model assumptions (\ref{eq:weights}). However, for time series data, we may be able to use different weights that penalize the distant lagged values more than the nearer lagged values. We describe such customization of adaptive weights in the next section.

\subsubsection{PaEBack with Adaptive Weights} \label{sec:aw}

In time series data, correlations between distant time lags typically exhibit a decreasing trend. Consequently, it is expected that the magnitudes of significant coefficients will diminish as the time lag increases. The adaptive weights, as defined in Equation (\ref{eq:weights}), are inversely related to the absolute values of the initial estimator. This enables us to generate non-decreasing adaptive weights by adjusting the absolute value of the initial estimator in a non-increasing order. Drawing inspiration from the principles of monotone regression, which involves fitting a monotone function to a set of data points in a plane (\cite{de2009isotone}), we have adopted a straightforward approach known as the PaEBack non-increasing adjusted adaptive weight algorithm (refer to Algorithm \ref{algo:wt} in Appendix \ref{sec:appC}). 




\section{Numerical Illustrations using Simulated Data}
\label{sec:sim}

This section provides comprehensive experimental results from our simulation analysis, focusing on investigating various order selection methods and their impact on the PaEBack framework.

We begin by presenting the simulation results under oracle settings in Section \ref{sec:ora}, aiming to gain insights into the PaEBack dual efficiency and optimal development size under varying horizons and historical sample sizes. Regarding the effects of order selection on PaEBack efficiency, we explore different models and analyze the findings in Section \ref{sec:ord}. Additionally, we conduct experiments to assess the effects of model misspecification by generating data from the threshold autoregressive (TAR) process and fitting it using AR models in Section \ref{sec:tar}, which allows us to examine the sensitivity of AR approximation under the PaEBack framework, as discussed in Section \ref{sec:tar}. Furthermore, we include a comparison with other subsampling methods in Appendix \ref{sec:fkc}.

\subsection{PaEBack Dual Efficiency under Oracle AR Setting for AR Process} \label{sec:ora}

Under the oracle setting, we fit the auto-regressive models with known true order $p$ but varying historical sample size $n$ and forecasting horizon $h$ to gain insights into predictive and sample efficiency performance under different scenarios. Simulated data is generated independently from the stationary AR(5) process with $\phi=c(0.5, -0.4, 0.3, -0.2, 0.1)^T$. We present the median of the efficiency curves for each $k$ based on 1000 replicates of the time series data in Figure \ref{fig:ora} and the baseline average $MSE(n)$ while using the full development sample in Table \ref{tab:ex1}.

\begin{table}[h]
	\begin{center}
		\caption{ The baseline average $MSE(n)$ (based on 1000 replicates), of using full development data in the oracle setting corresponding to the $MSE(n)$ of Figure \ref{fig:ora}. }
		\label{tab:ex1}
		\begin{tabular}{@{}c ccc@{}} \\ \toprule
			n & h = 3 & h = 5 & h = 10 \\ \midrule
			100 & 1.192 & 1.216 & 1.287 \\ 
			250 & 1.262 & 1.283 & 1.307 \\ 
			500 & 1.180 & 1.238 & 1.280 \\ 
			1000 & 1.222 & 1.239 & 1.264 \\  \bottomrule
		\end{tabular}
	\end{center}
\end{table}

Notice that when using full historical data, the forecasting error in Table \ref{tab:ex1} is the denominator of the predictive ratio in Figure \ref{fig:ora}. The larger $MSE(n)$ in Table \ref{tab:ex1}, the larger denominator of the y-axis in Figure \ref{fig:ora}, and the lower predictive ratio at the starting sample ratio. Under the aforementioned oracle setting, we could observe the impacts on dual efficiency from the following perspectives: 

\begin{itemize}
	\item[(i)] \textbf{Forecasting horizon $h$:} Given the true order is $p=5$, we consider the prediction step of lengths $h=3, 5, 10$ corresponding to the situation when $h$ is less than $p$, equal to $p$, or greater than $p$ respectively. As the results show, the predictive ratio for a larger horizon $h$ has an earlier trend to converge since the size of the validation set is equal to the forecasting horizon $h$, and the larger the horizon, the longer the validation set, and by the averaging effect, the more stable forecasts are given.
	
	\item[(ii)] \textbf{Historical sample size $n$:} The total length of historical data is varied from $n=100$, $250$, $500$, to $n= 1000$, which represents short, moderate, medium, and large sample sizes encountered in practice. For example, for daily stock prices, a year's worth of data would typically be of size $n=250$. We could observe that the predictive ratio $r_p(k)$ has a decreasing pattern, converging to 1 while the sample ratio $r_s(k)$ increases to 1. The larger the historical sample size $n$, the smaller the sample ratio $r_s(k)$ for the predictive ratio to converge to 1.

	\item[(iii)] \textbf{Practical Optimality:} Notice that there exists the probability of achieving better forecasting accuracy while using fewer development samples rather than the whole historical series. When $h=3$ and $n=1000$, using the most recent $30\% \sim 60\%$ observations can obtain even better forecasting performance than using the whole data.
\end{itemize}

\subsection{PaEBack Dual Efficiency with Order Selection for AR Process} \label{sec:ord}
In this section, we study the effects of not knowing the true order of an AR model under different order selection strategies. 1000 simulated data is generated independently from the stationary AR(5) process with $\phi=c(0.5, -0.4, 0.3, -0.2, 0.1)^T$ and the forecasting horizon is fixed at $h=5$. As introduced in Section \ref{sec:ordersel}, we have investigated the effects on dual efficiency with the following four order selection methods:
\begin{itemize}
    \item  \textbf{YW:} the oracle model using the Yule-Walker estimations with the known correct order $p$; 
    \item \textbf{AL:} the adaptive LASSO to select the order of the AR model using the estimator in Eq. (\ref{eq:adalasso}) with $\lambda$ selected by cross-validation using the PaEBack-SW approach, which we introduced in Section \ref{sec:221}.
    \item \textbf{AE:} the adaptive elastic net method to select the order of the AR model with $\lambda$ selected by cross-validation using the PaEBack-SW strategy, while parameter $\alpha$ defined in Eq.(\ref{eq:ae}) which determines the proportion between $\ell_1$ and $\ell_2$ penalty is set to be fixed at $0.5$.
    \item \textbf{ATE:} the adaptive elastic net method to select the order of the AR model both parameters $\lambda$ and $\alpha$ tuned by cross-validation using the PaEBack-SW strategy.
\end{itemize}
 
\begin{figure}[h]
	\centering
	\begin{minipage}[t]{.95\linewidth}
		\centering
		\includegraphics[width=\textwidth]{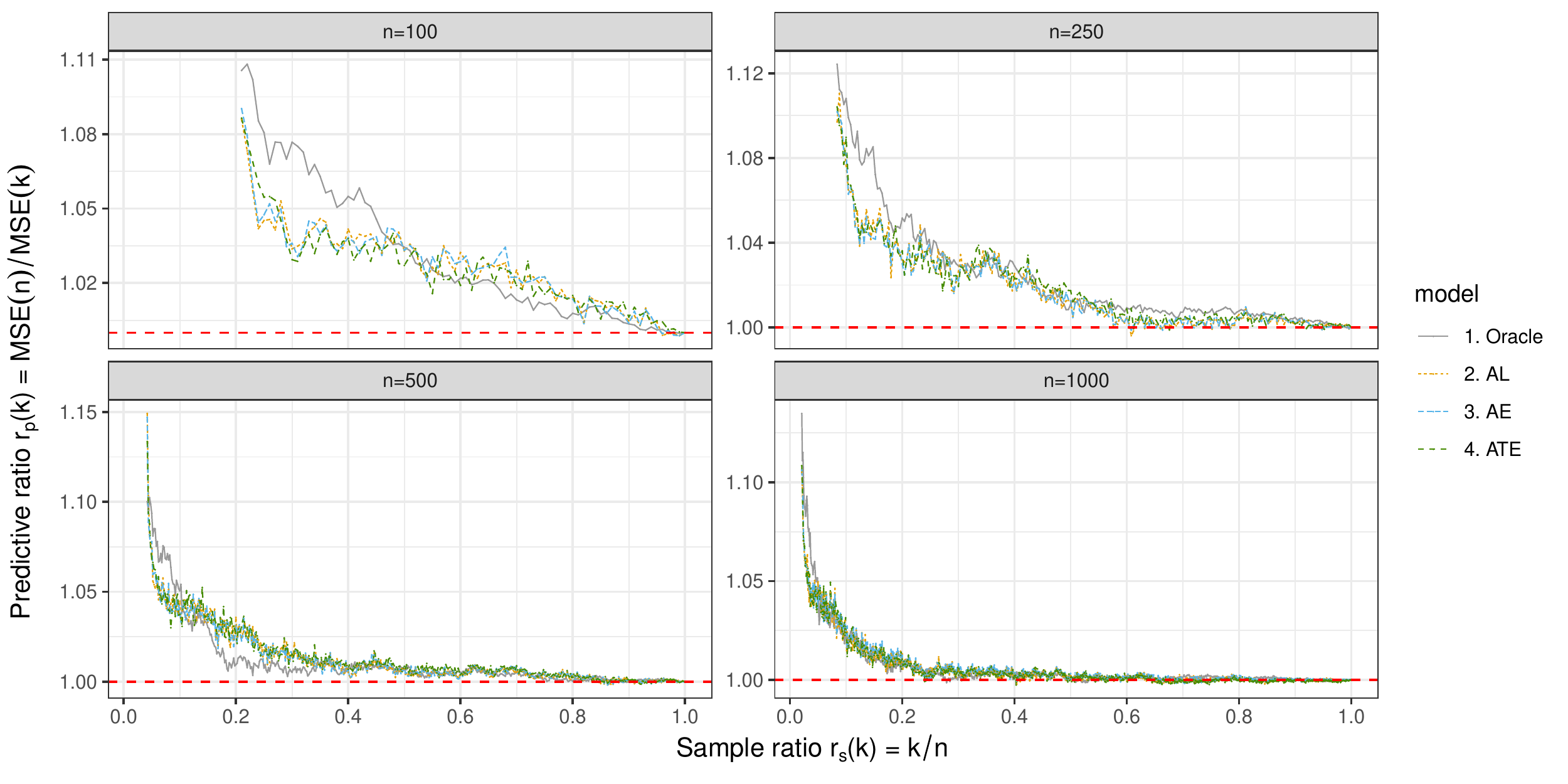}
	\end{minipage}
	\centering
	\caption{Two-fold efficiency of different order selection methods. The development size is varied from $n=100$, $250$, $500$, to $n=1000$ with forecasting horizon fixed at $h=5$ using four different order selection methods described at the beginning of Section \ref{sec:sim}: (1) Oracle (grey solid); (2) ALASSO (orange dotted); (3) Adaptive elastic net (blue dashed); (4) adaptive elastic net with both parameters tuned by cross-validation (wider green dashed).
		Corresponding $MSE(n)$s, the denominator of the y-axis, are displayed in Table \ref{tab:simu2}.}
	\label{fig:simu2}
\end{figure}

\begin{table}[h]
	\begin{center}
		\caption{The baseline average $MSE(n)$ (based on 1000 replicates), of using full development data in studying order selection effects in Section \ref{sec:ord} corresponding to $MSE(n)$ of Figure \ref{fig:simu2}.}
		\label{tab:simu2}
		\begin{tabular}{@{}c cccc@{}} \\ \toprule
			n& YW & AL & AE & ATE \\  \midrule
			100 & {1.2161} & 1.2416 & 1.2422 & 1.2413 \\ 
			250 & {1.2833} & 1.2960 & 1.2961 & 1.2939 \\ 
			500 & {1.2380} & 1.2445 & 1.2450 & 1.2439 \\ 
			1000 & {1.2385} & 1.2408 & 1.2407 & 1.2418 \\   \bottomrule
		\end{tabular}
	\end{center}
\end{table}

As Figure \ref{fig:simu2} shows, the main trend of predictive ratio curves of all order selection methods is similar: predictive ratio decreases while sampling ratio increases and converges to 1 as the $r_s(k) \rightarrow 1$. However, when the sample ratio is small, the relative predictive ratio of using penalized variable selection methods drops faster so that by adding a small percentage of the development sample, the forecasting performance of order selection methods improves more rapidly than the oracle setting. Nevertheless, this improvement is not that substantial when the development sample size is large. Notice that the denominator of the oracle model (YW), as Table \ref{tab:simu2} shows, is lower than that of all the other methods since it assumes the true model and has a relatively low forecasting error to shrink.

\subsection{PaEBack Dual Efficiency under Model Misspecification} \label{sec:tar} 

In this section, we investigate the effect of model misspecification on the dual efficiency of the PaEBack framework, corresponding to real-life scenarios where the fitted models are seldom perfect or even close. There have been studies on the robustness of fitting a $p^{th}$ order AR model in stationary linear time series models without knowing the actual order of an AR process (\cite{bhansali1981effects, kunitomo1985properties}). 

The threshold autoregressive (TAR) example has been used in illustrating model misspecification for subsampling methods (\cite{fukuchi1999subsampling}). Following their set-up, we generate time series data from TAR($1$) of the form:
\begin{equation} \label{eq:tar}
X_t=\left\{ \begin{array}{cc}
0.14+0.10X_{t-1}+\epsilon_t & \text{if } X_{t-1} < -0.2 ,\\
0.80 X_{t-1} + \epsilon_t & \text{if } X_{t-1} \geq -0.2 ,
\end{array} \right.
\end{equation}
We apply the AR approximation and order selection technique introduced in the earlier section to these TAR time series data with the PaEBack technique to investigate the effects under such misspecification scenarios.

\begin{figure}[h]
	\centering
	\begin{minipage}[t]{.99\linewidth}
		\centering
		\includegraphics[width=\textwidth]{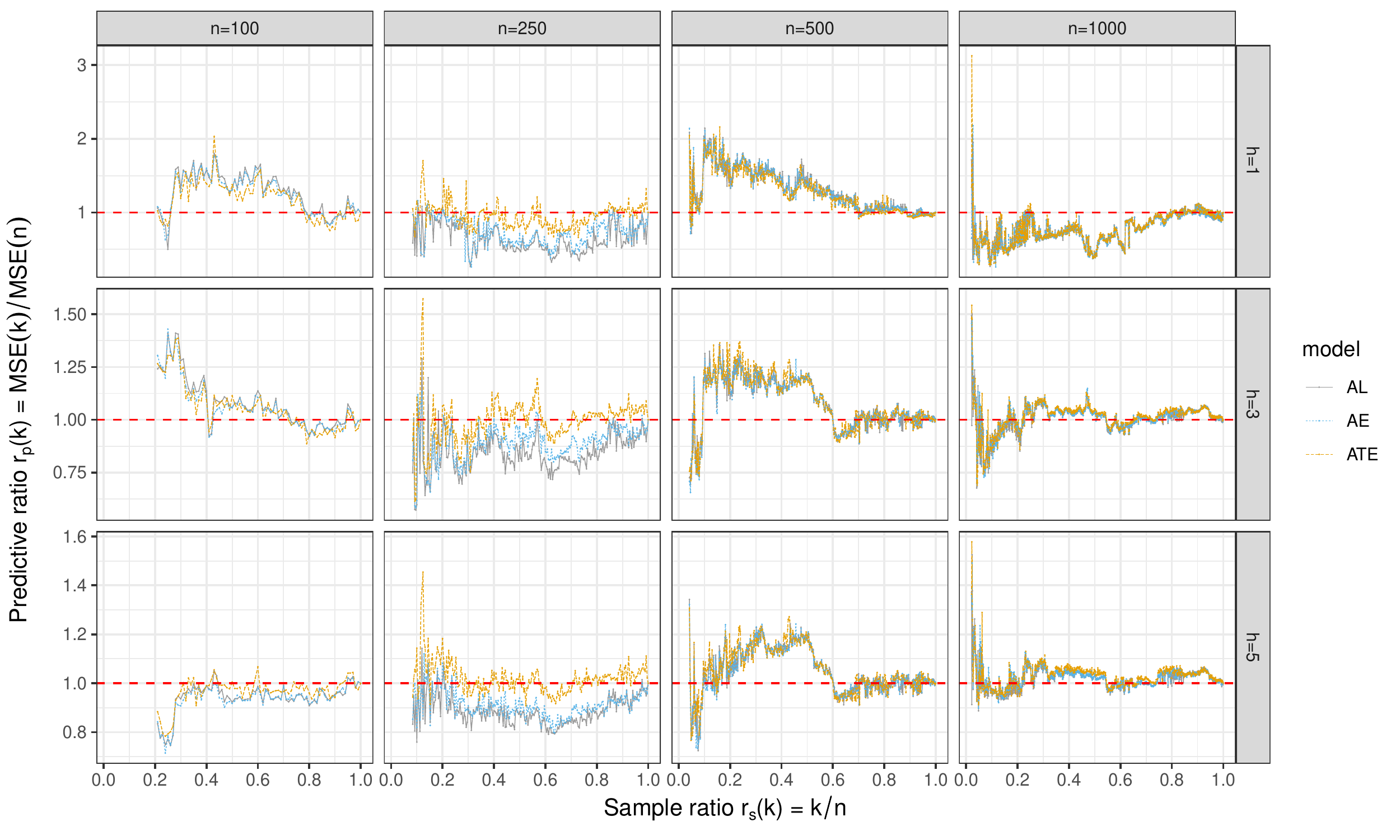}
	\end{minipage}%
	\centering
	\caption{PaEBack dual efficiency of misspecified AR models under simulated TAR data. We use three methods to estimate the order: ALASSO (grey solid);  adaptive elastic net (blue dotted); tuned adaptive elastic net (orange dashed ). 1000 simulated data is generated independently from the TAR(1) process defined in Eq. (\ref{eq:tar}) with historical sample size $n$ and forecast horizon $h$ varying as the labels show.}
	\label{fig:tar}
\end{figure}

\begin{table}[h]
	\begin{center}
		\caption{ The baseline average $MSE(n)$ of using the full historical data generated from TAR(1) process defined in Eq. (\ref{eq:tar}), corresponding to visualized results in Figure \ref{fig:tar}.}
		\label{tab:tar}
		\begin{tabular}{@{}c ccc c ccc@{}} \toprule
			& \multicolumn{3}{c}{$n=100$} & \phantom{abc}& \multicolumn{3}{c}{$n=250$} \\ \cmidrule{2-4} \cmidrule{6-8} 
			& AL & AE & ATE && AL & AE & ATE \\ \midrule
			h=1 & 1.103 & 1.104 & 1.157 & & 0.763 & 0.745 & 0.608 \\ 
			h=3 & 1.468 & 1.466 & 1.501 & & 1.284 & 1.221 & 1.120 \\ 
			h=5 & 1.653 & 1.653 & 1.642 & & 1.403 & 1.405 & 1.281 \\   \midrule
			& \multicolumn{3}{c}{$n=500$} & \phantom{abc}& \multicolumn{3}{c}{$n=1000$} \\ \cmidrule{2-4} \cmidrule{6-8} 
			& AL & AE & ATE & & AL & AE & ATE \\ \midrule
			h=1 & 1.521 & 1.519 & 1.536 & & 0.285 & 0.286 & 0.282 \\ 
			h=3 & 1.382 & 1.380 & 1.376 & & 0.556 & 0.557 & 0.553 \\ 
			h=5 & 1.058 & 1.059 & 1.047 & & 0.806 & 0.807 & 0.804 \\  \bottomrule
		\end{tabular}
	\end{center}
\end{table}

As Table \ref{tab:tar} shows, the forecasting error MSE$(n)$ for larger historical sample size $n$ is smaller. The adaptively tuned elastic net method provides better forecasting accuracy when $n > 100$. As Figure \ref{fig:tar} shows, most simulations have better prediction performances ($y<1$) when the sample ratio is small, for which we believe that the Pareto dual efficiency can be improved by adopting the PaEBack framework even under the model misspecification scenario.

\section{PaEBack: Application to Real Data Sets}
\label{sec:dat}

In this section, we demonstrate the practical application of PaEBack using well-known publicly available stock price data, offering a concrete example for easy understanding. Furthermore, we employ the PaEBack method to analyze confirmed cases of Coronavirus data, which exhibits high volatility. This case study enables us to evaluate the effectiveness of the proposed framework on non-stationary time series data, incorporating comprehensive model comparisons. Both examples serve to showcase the efficacy of the PaEBack method.

\subsection{Log Return of Stock Prices} \label{sec:stock}

In financial analysis, stock prices are typical time series data. Forecasting the return of stock prices is one of the most effective tools for risk management. It has convenient publicly available resources that are easy to learn and reproduce (\cite{stock}).

Define the stock price at time $t$ as $p_t$, compared with the return $r_t = \frac{p_t - p_{t-1}}{p_{t-1}}$, which captures the relative difference of stock prices at time $t$, the log return $X_t$ is defined as
\begin{equation}
X_t = ln(1+r_t) = ln\left( \frac{p_t}{p_{t-1}}  \right) = ln(p_t) - ln(p_{t-1}),
\end{equation}
whose following properties may be useful to establish better statistical analysis: (1) The time-additive property of $ln\left( \prod_{t=1}^n (1+r_t) \right) = \sum_{t=1}^nX_t= ln(p_n) - ln(p_0)$ makes it easier to compute and preserves consistency when $n$ is large. (2) The normality assumption is easier to fit into the format of $X_t$. (3) When $r$ is small, $ln(1+r) \approx r$ can give a fairly accurate approximation. More descriptions on log return can be found in \cite{paparoditis2009resampling}.

Here, we take the end-of-day adjusted prices as an example from Yahoo Finance (\cite{stock}). The stock prices of four corporations used in the analysis from January 3, 2017 (n=1), to October 1, 2021 (n=1000), have been plotted in Figure \ref{fig:stockprice}.

\begin{figure}[h]
	\centering
	\begin{minipage}[t]{\linewidth}
		\centering
		\includegraphics[width=\textwidth]{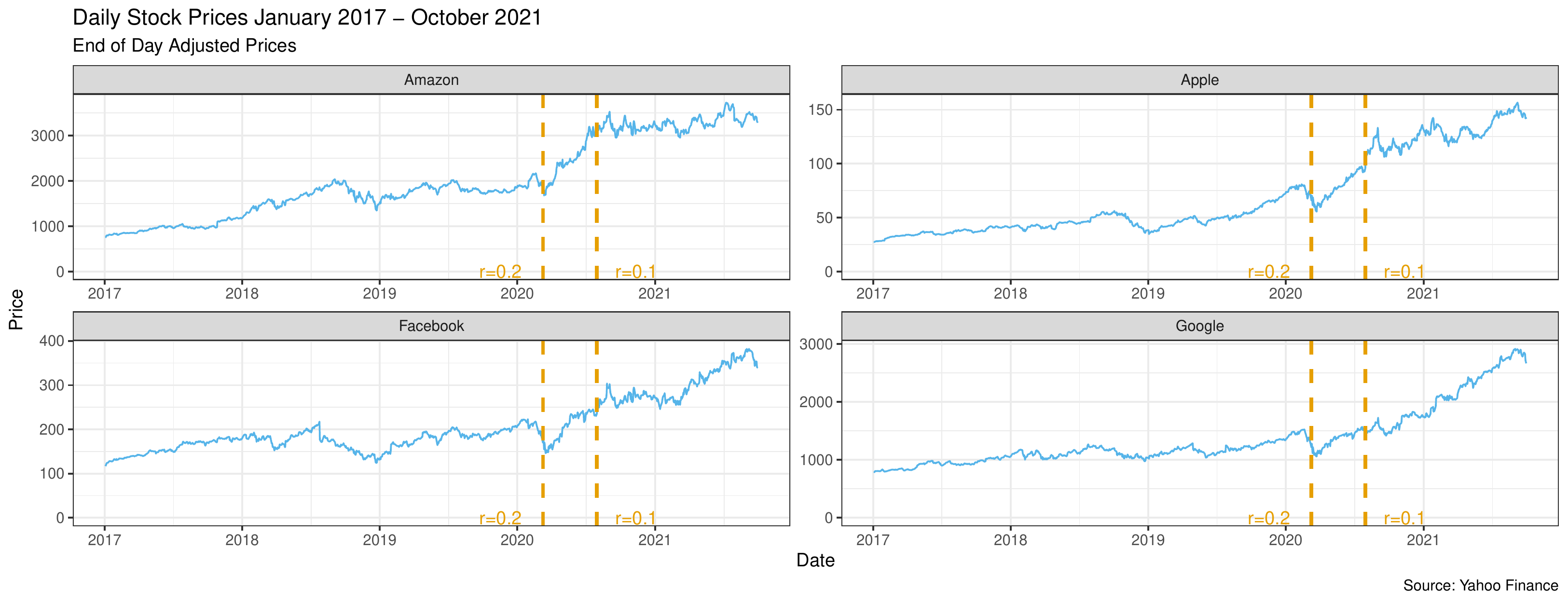}
	\end{minipage}%
	\centering
	\caption{The daily stock price of four corporations was used in the real data analysis.}
	\label{fig:stockprice}
\end{figure}


\subsubsection{Autoregressive Modeling}

\begin{figure}[h]
	\centering
	\begin{minipage}[t]{\linewidth}
		\centering
    \includegraphics[width=\textwidth]{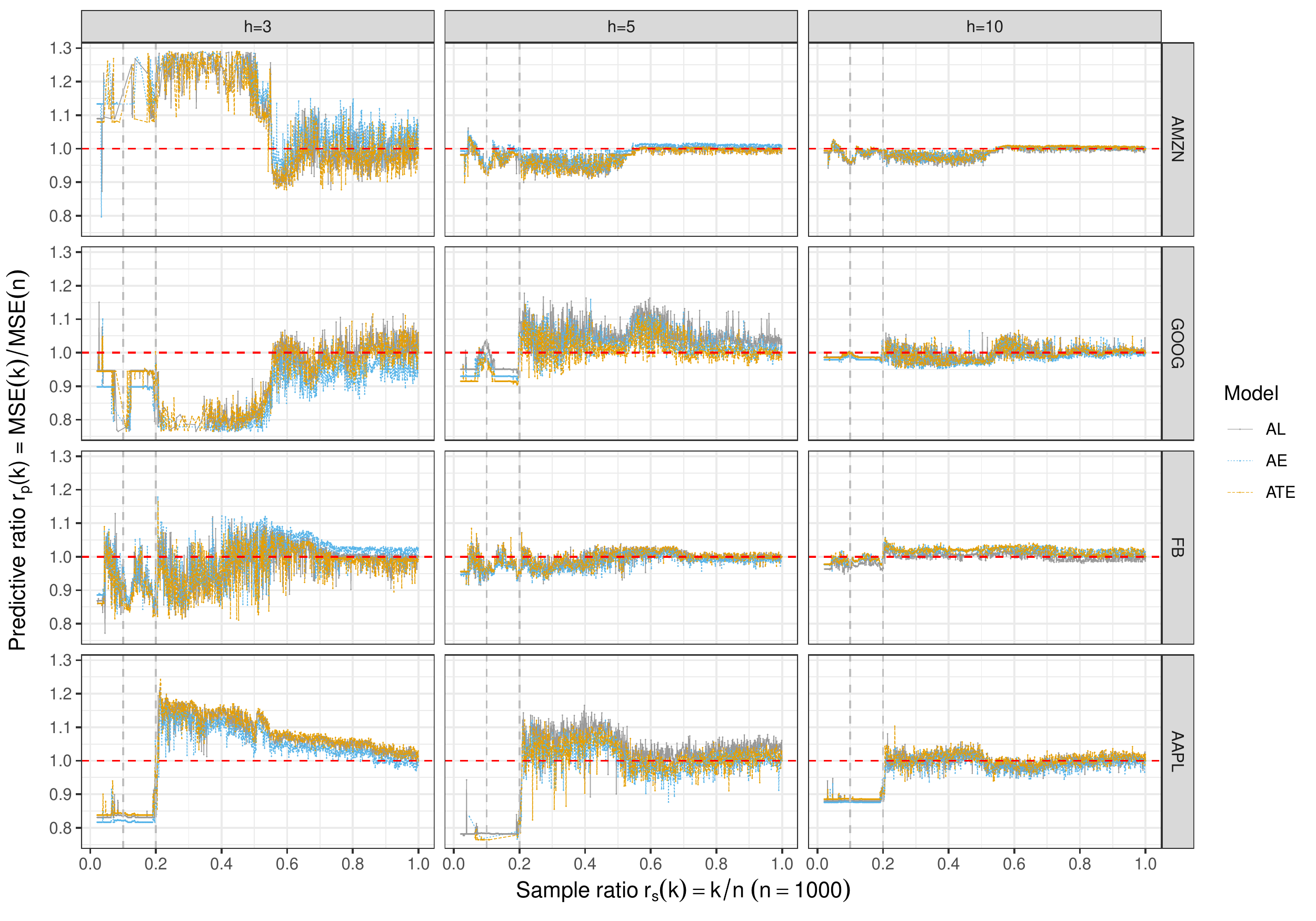}
	\end{minipage}%
	\centering
	\caption{PaEBack efficiency for the short-term stock price forecast with varying development size $n$, model, forecasting horizon $h$ of four different corporations. \label{fig:stockresult}}
\end{figure}

We set the forecasting horizon at $h=3, 5, 10$ respectively, and historical sample size varying as $n=100,250, 500, 1000$ with corresponding dates. We use different methods to select the order and record the $MSE(n)$  in Table \ref{tab:stock} with visualized PaEBack efficiency in Figure \ref{fig:stockresult}.

\begin{table}[h]
	\begin{center}
		\caption{ The forecasting error, $MSE(n) \times 10^{-4}$, of using $n=1000$ historical observations of  for each corporations.}
		\label{tab:stock}
		\begin{tabular}{@{}l ccc c ccc@{}} \toprule
			& \multicolumn{3}{c}{$Apple$} & \phantom{abc}& \multicolumn{3}{c}{$Amazon$} \\ \cmidrule{2-4} \cmidrule{6-8} 
			& h = 3 & h = 5 & h = 10 & & h = 3 & h = 5 & h = 10 \\ \midrule
			AL & 3.684 & 6.078 & 5.005 & & 0.182 & 2.897 & 2.840 \\ 
			AE & 3.749 & 6.269 & 5.029 & & 0.175 & 2.867 & 2.824 \\ 
			ATE & 3.654 & 6.233 & 4.974 & & 0.184 & 2.904 & 2.823 \\   \midrule
			& \multicolumn{3}{c}{$Facebook$} & \phantom{abc}& \multicolumn{3}{c}{$Google$} \\ \cmidrule{2-4} \cmidrule{6-8} 
			& h = 3 & h = 5 & h = 10 & & h = 3 & h = 5 & h = 10 \\ \midrule
			AL & 1.758 & 3.663 & 3.301  & & 0.436 & 1.427 & 1.109 \\ 
			AE & 1.724 & 3.689 & 3.256 & & 0.460 & 1.459 & 1.118 \\ 
			ATE & 1.775 & 3.657 & 3.252 & &  0.437 & 1.482 & 1.110 \\   \bottomrule
		\end{tabular}
	\end{center}
\end{table}

According to Table \ref{tab:stock} and Figure \ref{fig:stockresult}, we have the following observations:
(1) The predictive ratio is low and even less than 1 for small sample ratios, which validates the statement that using fewer development samples but closer observations can practically improve the predictive accuracy aside from saving computational time.
(2) There are two noticeable changes in the predictive efficiency at sample ratio around $r_s(k)=0.1$ and $r_s(k)=0.2$, corresponding to development sample size $k \approx 100$ and $k \approx 200$, and the date around March 3, 2020, and July 30, 2020, labeled in Figure \ref{fig:stockprice}, the daily stock price plot as well.
These two points match two critical change points of the phase of Coronavirus. On March 11, 2020, the World Health Organization (WHO) declared COVID-19 a global pandemic (\cite{cucinotta2020declares}) when the pandemic started to impact financial marketing. At the end of July, multiple biotech giants such as Pfizer, BioNTech, and Moderna announced promising clinical results with funding. After several days, on August 3, a new pandemic phase was officially announced. The evidence strongly supports the current model's ability to capture the time series character using PaEBack samples.

\subsubsection{Model Fitting Using Machine Learning Methods} \label{sec:stockm}
Stock price forecasting models mainly have two types: traditional time series and machine learning methods. On stock data during COVID-19, the performance of the machine learning models (Long-Short Term Memory and XGBoost) was not as good as the AR models or the Last Value models (\cite{mottaghi2021stock}). This might be due to a strong correlation between the price values of close days, for which the PaEBack framework can help to improve dual efficiency. Among traditional time series models to forecast stock price, a combination of the ARIMA and the generalized autoregressive conditional heteroskedasticity (GARCH) model has been shown to yield better performances (\cite{grachev2017application, gao2021research}). Hence, we compare the performance of the following seven models using either the full historical sample size or their PaEBack subsamples.

\begin{itemize}
    \item  \textbf{AR:} The best traditional AR model according to the Akaike information criterion (AIC).
    \item  \textbf{ARLasso:} The adaptive Lasso is introduced in Section \ref{sec:ordersel}.
    \item \textbf{ARElas:} The proposed adaptive elastic-net estimator using adaptive weights with PaEBack-SW parameter selection is introduced in Section \ref{sec:ordersel}.
    \item  \textbf{ARIMA:} The ARIMA(p,d,q) model with step-wise parameter selection method in \cite{hyndman2008automatic}.
    \item  \textbf{GARCH:} Among all GARCH($p,q$) models with varying $p$'s and $q$'s, GARCH (1,1) has been found to perform well in forecasting stock price, while higher parameters overestimate the levels of volatility (\cite{grachev2017application, gao2021research}). Hence, the performance of the GARCH(p,q) model with a combination of low orders of $p$'s and $q$'s are compared based on the Bayesian information criterion (BIC) (\cite{schwarz1978estimating}), which is a widely applied model selection criterion. As the \textit{"GARCH"} column of Table \ref{tab:garchbic} in Appendix \ref{sec:appD} shows, GARCH(1,1) yields the least BIC, and is thus selected.
    \item  \textbf{ARIMA-GARCH (AGARCH):} The combination of ARIMA(0,0,1)-GARCH(1,1) is determined based on the following steps.  (a) Fix GARCH(1,1) as explained in (v). (b) Since the first-order differencing is sufficient to model the original stock price (\cite{ariyo2014stock}) and the log return has already calculated the difference, the differencing parameter $d$ of ARIMA(p,d,q) is thus fixed at $d=0$. (c) Among the ARIMA(p,0,q)-GARCH(1,1) variants suggested in \cite{grachev2017application}, $p=0$ and $q=1$ yields the least BIC as the  \textit{"AGARCH"} column in Table \ref{tab:garchbic} in Appendix \ref{sec:appD} shows.
     
    \item  \textbf{gjrGARCH:} The gjrGARCH models improve GARCH by modeling the positive and negative values  asymmetrically (\cite{glosten1993relation}). Among gjrGARCH(p,q), $p=0$ and $q=1$ is selected according to the BIC shown in the \textit{"gjrGARCH"} column of Table \ref{tab:garchbic} in Appendix \ref{sec:appD}.
    \item \textbf{ARIMA-gjrGARCH(AgjrGARCH):}
    The combination of ARIMA(0,0,1) and gjrGARCH(1,1) is determined by firstly fixing ARIMA(p,0,q) and gjrGARCH(1,1) and selecting $p$ and $q$ based on the BIC shown in the \textit{"AgjrGARCH"} column of Table \ref{tab:garchbic} in Appendix \ref{sec:appD}.
\end{itemize}


Here we take the $h=10$ days-ahead forecast on the log return of Amazon's stock price with $n=1000$ historical data points as an example to compare the practical performance of these eight models. Figure \ref{fig:amzn} illustrate their corresponding Pareto optimal efficiency with the y-axis as $MSE(n)$ $MSE(k)=r_p(k)*MSE(n)$. By multiplying the predictive ratio with $MSE(n)$, it is more convenient to compare the predictive performance of different models by simply visualizing the curve's height. Notice that the red dotted line is $MSE(n)$. The curve above the red dotted line means a $r_p(k)>1$ while the curve under the red dotted line means a practical predictive ratio $r_p(k)<1$, indicating a better practical predictive performance of using less sample size.

\begin{figure}[h]
	\centering
	\begin{minipage}[t]{\linewidth}
		\centering
		\includegraphics[width=\textwidth]{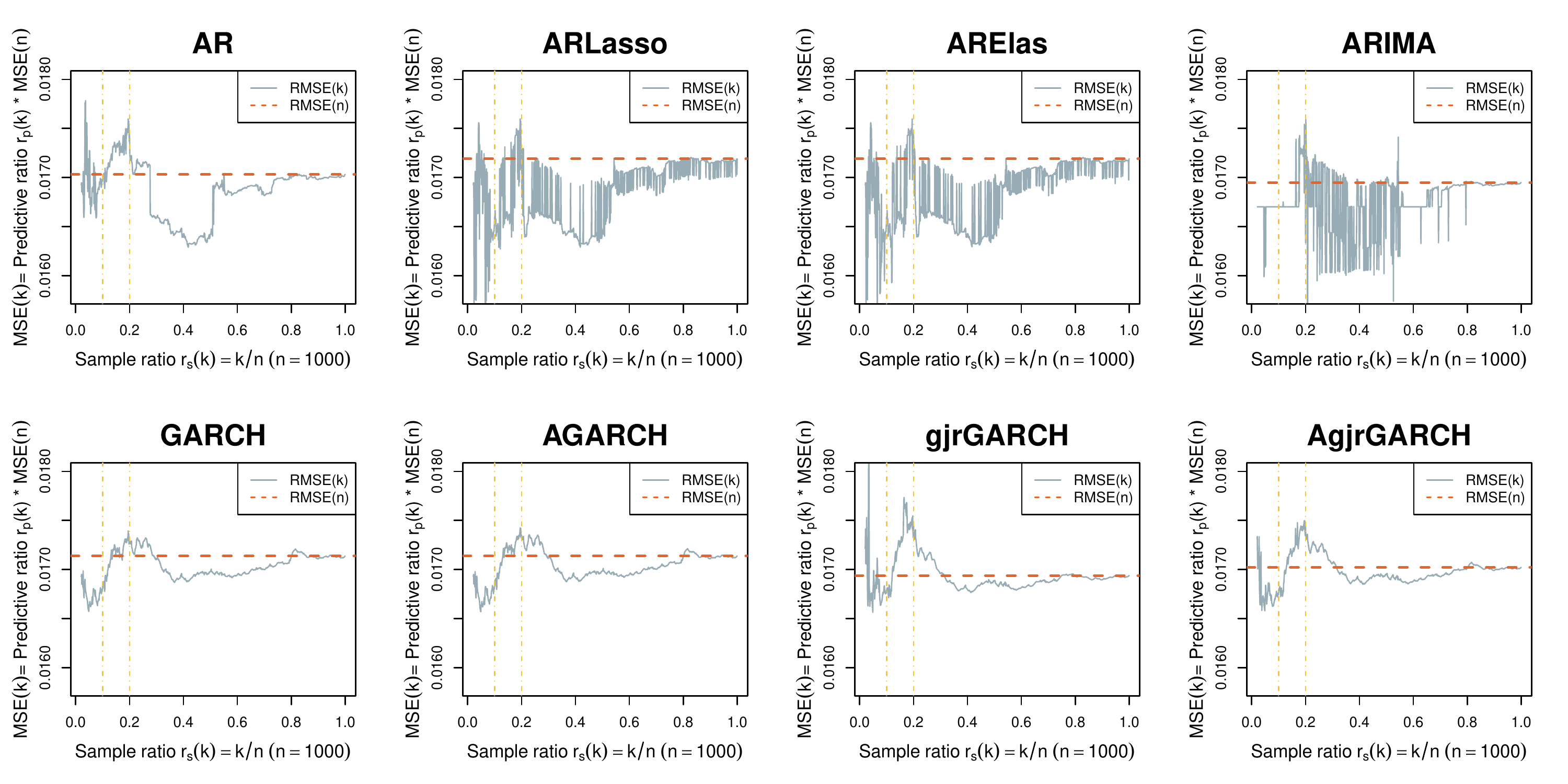}
	\end{minipage}%
	\centering
	\caption{The Pareto efficiency plot of $h=10$ days-ahead forecasting of Amazon's daily stock return. The predictive performances of eight methods introduced in Section \ref{sec:stockm} are evaluated by RMSE. The horizontal red line is the baseline RMSE using a full historical sample size $n=1000$, while the vertical orange lines characterize two critical ratios.}
	\label{fig:amzn}
\end{figure}

As Figure \ref{fig:amzn} shows, we can observe: 
(1) The predictive efficiency $r_p(k)$ drastically changes, decreases or even increases when $r_s(k)$ is small but converges to 1 when the sample ratio $r_s(k)$ is larger than some threshold. Hence, the improvement in the practical performance using more samples is less and less efficient as the sample ratio increases. (2) The optimal practical relevant predictive ratio $r_p(k^{opt})$ for all methods is less than 1, i.e., the practical predictive performance of using fewer samples yields even better predictive performance than using the complete historical data. 

\begin{table}[h]
\centering
\caption{Optimal PaEBack sample size and dual efficiency for the stock return data. Here, the optimal PaEBack sample size k is selected based on forecasting error evaluated by the test RMSE. \label{tab:stockres}}
\begin{tabular}{@{}ccccccc@{}}
  \hline \toprule
     \multicolumn{1}{c}{\multirow{2}{*}{$Methods$}}& & \multicolumn{2}{c}{$RMSE (\times 10^{-2})$} & \phantom{a}&   \multicolumn{2}{c}{Dual Efficiency} \\ \cmidrule{3-4} \cmidrule{6-7} 
 & $k^{opt}$ & $RMSE(k^{opt})$ & $RMSE(n)$ & & $r_p(k^{opt})$ & $r_s(k^{opt})$ \\ \midrule
AR &    418 &   1.62898 & 1.70313 &  & 0.95646 & 0.418 \\ 
  ARLasso &     24 &   1.56011 & 1.71907 &  & 0.90753 & \textbf{0.024} \\ 
  ARElas &     81 &    1.63577 & 1.70517 &  & 0.95931 & 0.081 \\ 
  ARIMA &    208 &   \textbf{1.51133} & 1.69494 &  & \textbf{0.89167} & 0.208 \\ 
  GARCH &     49 &   1.65728 & 1.71379 &  & 0.96703 & 0.049 \\ 
  AGARCH &     49 &   1.65699 & 1.71386 &  & 0.96681 & 0.049 \\ 
  gjrGARCH &     49 &   1.65645 & \textbf{1.69384} &  & 0.97793 & 0.049 \\ 
  AgjrGARCH &     49 &   1.65796 & 1.70225 &  & 0.97398 & 0.049 \\  \bottomrule
\end{tabular}
\end{table}

In addition to visualization, Table \ref{tab:stockres} displays the exact practical optimal PaEBack sample size $k^{opt}$ and corresponding dual efficiency. As shown in Table \ref{tab:stockres}: (1) The practical predictive accuracy can be improved for all methods using less sample size $k^{opt}$. (2) When using the full sample, the combination of the ARIMA-GARCH model yields the best predictive performance with $RMSE = 1.6938$. However, all methods using PaEBack optimal sample size $k^{opt}$ achieve better $RMSE$ than 1.6938, demonstrating the PaEBack framework's effectiveness. (3) Looking at $r_p(k^{opt})$ which demonstrates the relative predictive ratio, we could see higher $r_p(k^{opt})$ for GARCH-related methods which assume complex features while the $r_p(k^{opt})$ is lower for simpler methods.

\subsection{Nowcasting of COVID-19 Confirmed Cases} \label{sec:covid}

Forecasting confirmed cases of Coronavirus has been challenging due to its extreme uncertainty and non-stationarity. Hence, we apply the PaEBack framework here and investigate its performance on twenty statistical and machine learning forecasting models as suggested in \cite{chakraborty2022nowcasting} and listed in Figure \ref{fig:tsmet}.

\begin{figure}[h]
    \centering
    \begin{minipage}[t]{\linewidth}
    \centering
    \includegraphics[width=.9\textwidth]{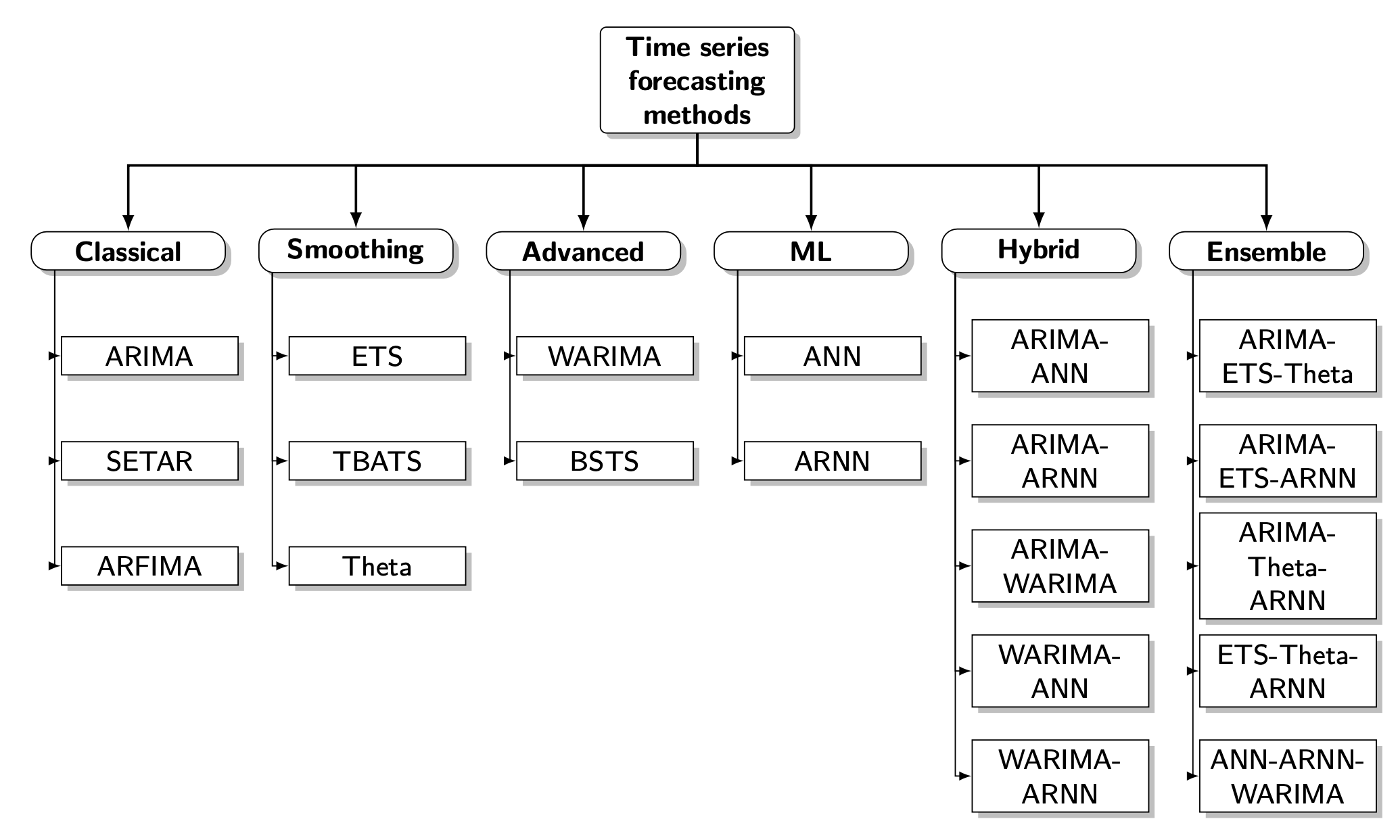}
    \caption{A systemic view of the various forecasting methods used in \cite{chakraborty2022nowcasting} with a detailed description in Appendix \ref{sec:covidsupp}. 
}.
    \label{fig:tsmet}
    \end{minipage}
\end{figure}

Due to the practical need to provide reasonable suggestions to policymakers, we replicated the exact analysis on the complete training set as in \cite{chakraborty2022nowcasting} and presented both the RMSE and the SMAPE for the 30-day ahead forecast in Table \ref{tab:covid}. The models are trained on the exact USA data with a 210-day complete history from Jan 20, 2020, to Aug 16, 2020. The forecasting performances are compared on the same test set with $h=30$ from Aug 17, 2020, to Sep 15, 2020. The "full" column represents the result using the complete training set with a 210-day history. The optimal development sample size $k$ is selected among $\{61, 62, \cdots, 210\}$ using the PaEBack framework.

Similarly to Figure \ref{fig:stockresult}, we multiply the predictive ratio with a constant $MSE(n)$ for each model on the y-axis to make the results visually comparable in their predictive performance. Still, the curve above the red dotted line means $r_p(k)>1$, with $MSE(k)>MSE(n)$. In contrast, the curve under the red dotted line represents $r_p(k)<1$, indicating better practical predictive performance when using fewer samples. Figure \ref{fig:covid20} visualizes the dual PaEBack efficiency of twenty methods in forecasting the confirmed case of Coronavirus in the US with a $h=30$ horizon.

\begin{figure}[h] 
     \centering
     \begin{subfigure}[b]{0.58\textwidth}
         \centering
        \centering
        \includegraphics[width=\textwidth]{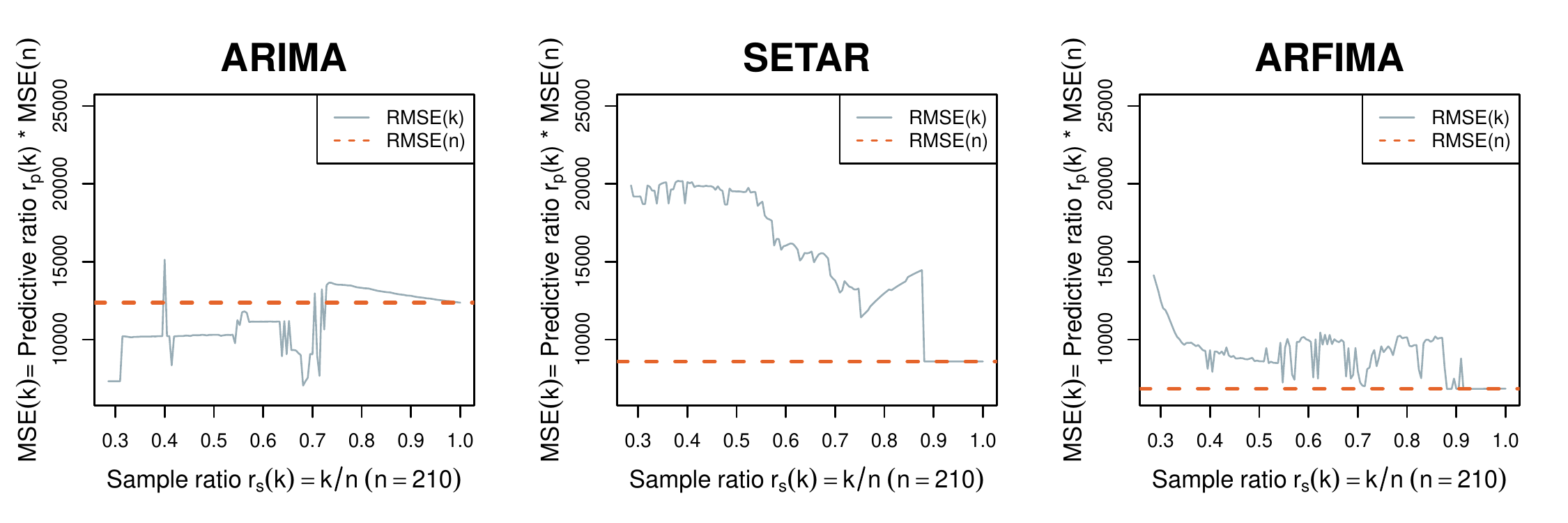}
        \caption{Classical}
        \label{fig:classic}
     \end{subfigure}
     \hfill
     \begin{subfigure}[b]{0.4\textwidth}
        \centering
        \includegraphics[width=\textwidth]{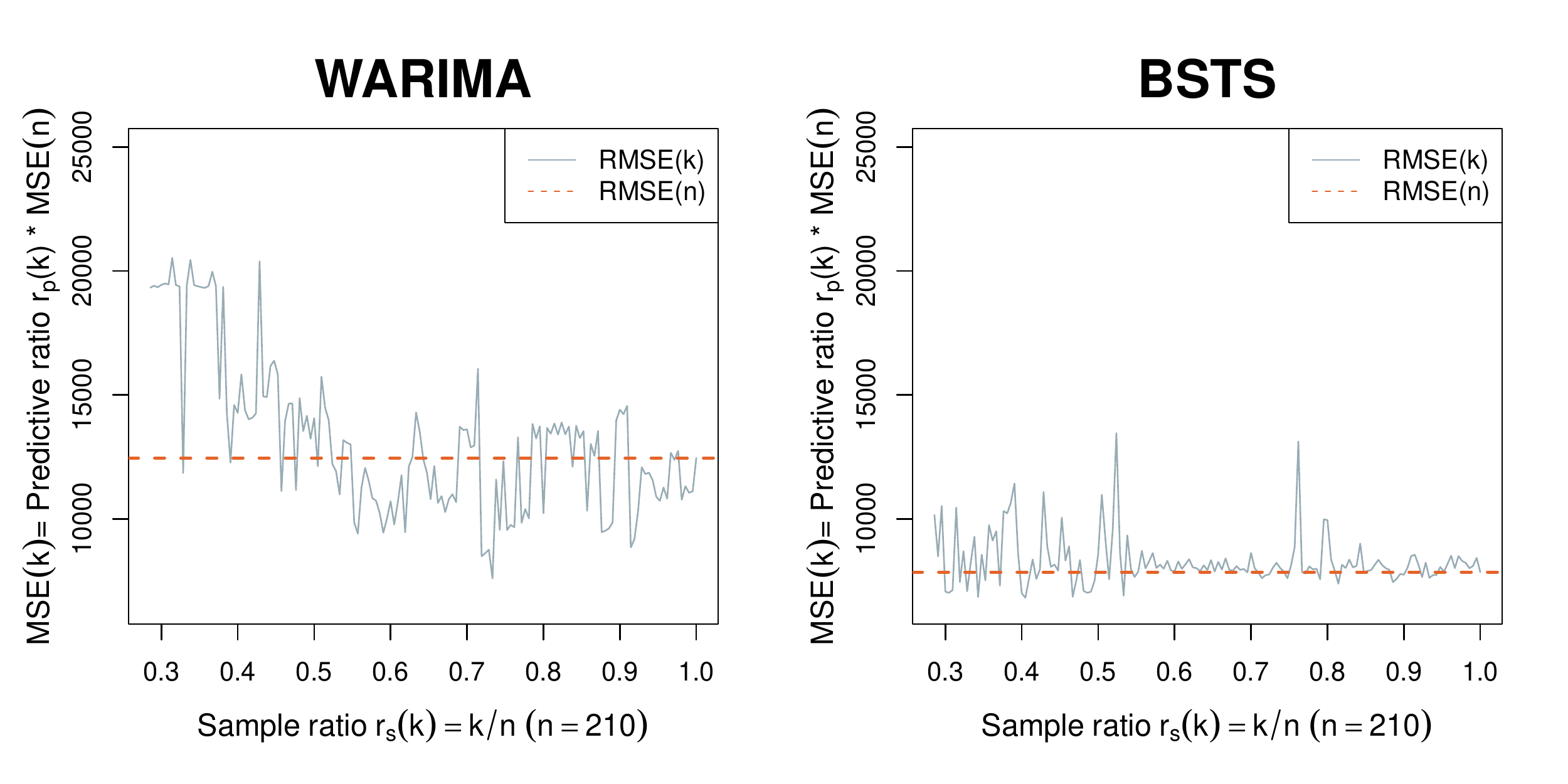}
        \caption{Advanced}
        \label{fig:advanced}
     \end{subfigure}
     \hfill
     \begin{subfigure}[b]{0.58\textwidth}
         \centering
        \centering
        \includegraphics[width=\textwidth]{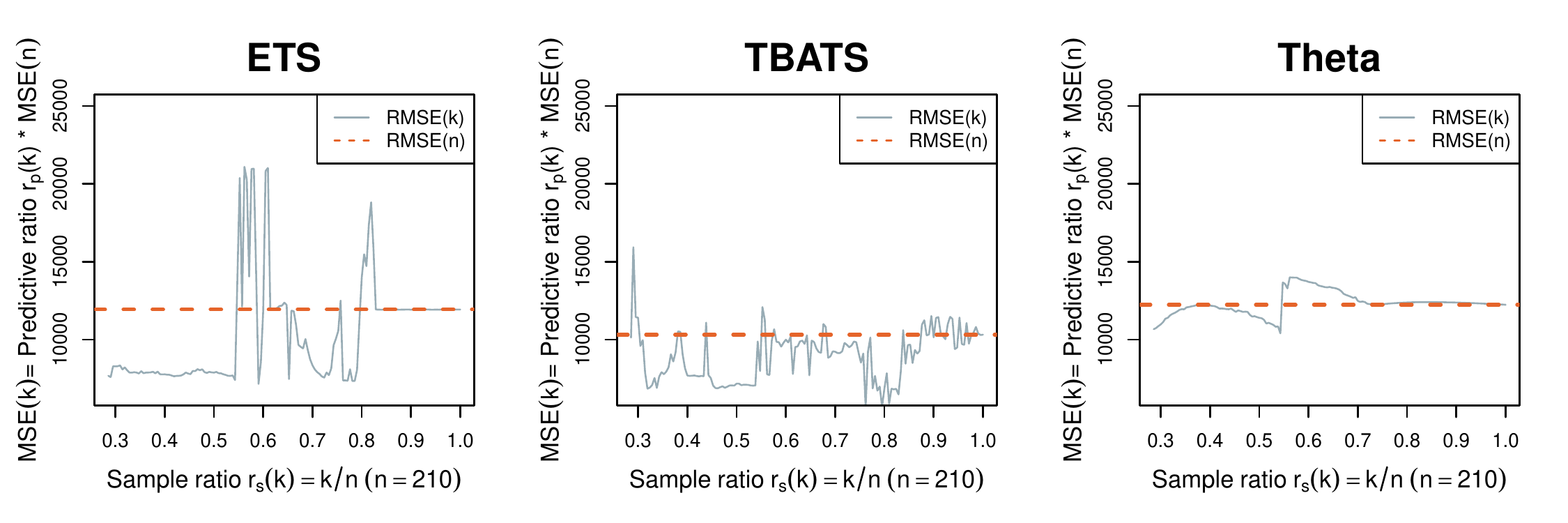}
        \caption{Smoothing}
        \label{fig:smooth}
     \end{subfigure}
     \begin{subfigure}[b]{0.4\textwidth}
         \centering
        \centering
        \includegraphics[width=\textwidth]{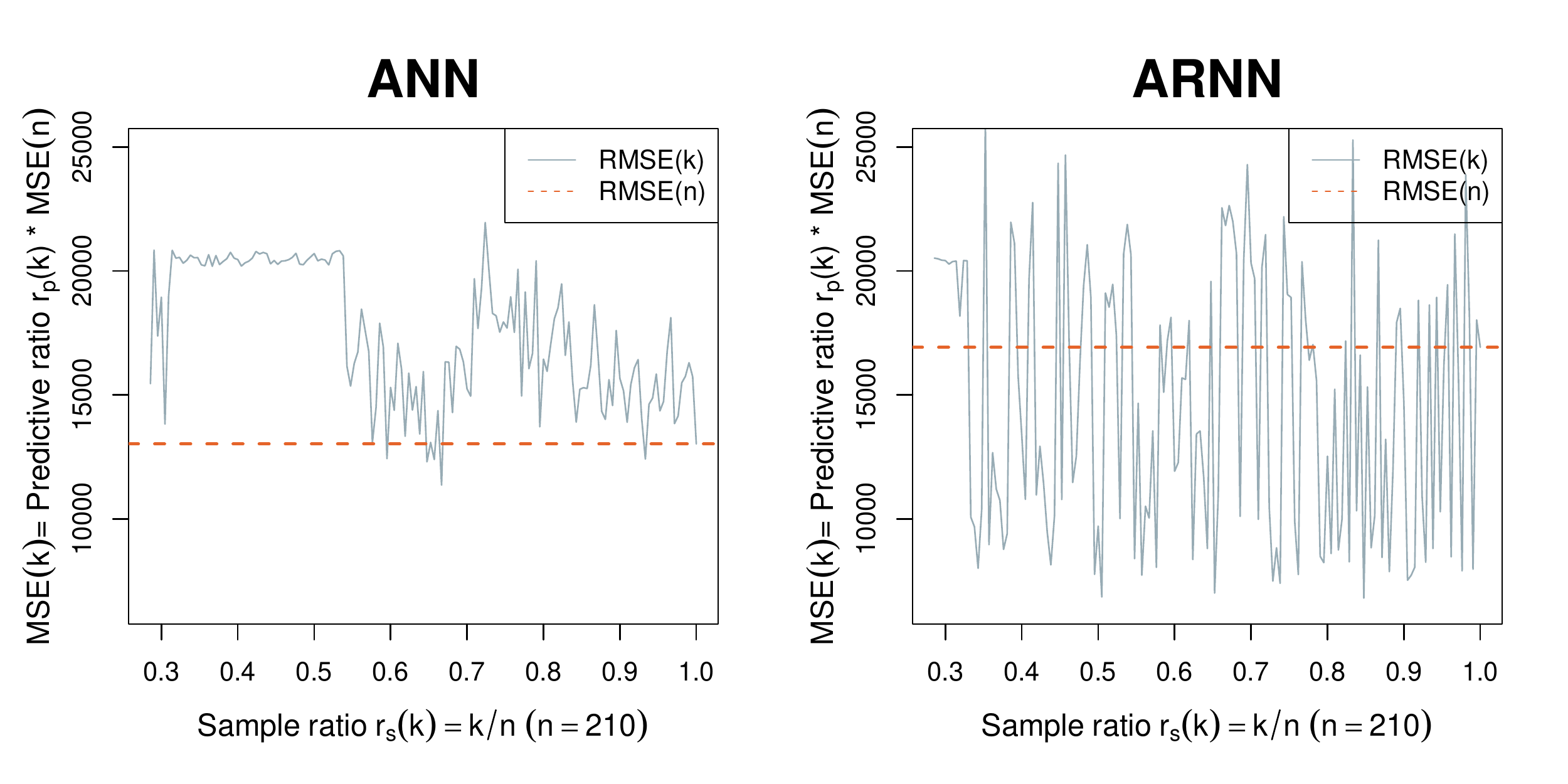}
        \caption{ML}
        \label{fig:ml}
     \end{subfigure}
     \begin{subfigure}[b]{0.95\textwidth}
         \centering
        \centering
        \includegraphics[width=\textwidth]{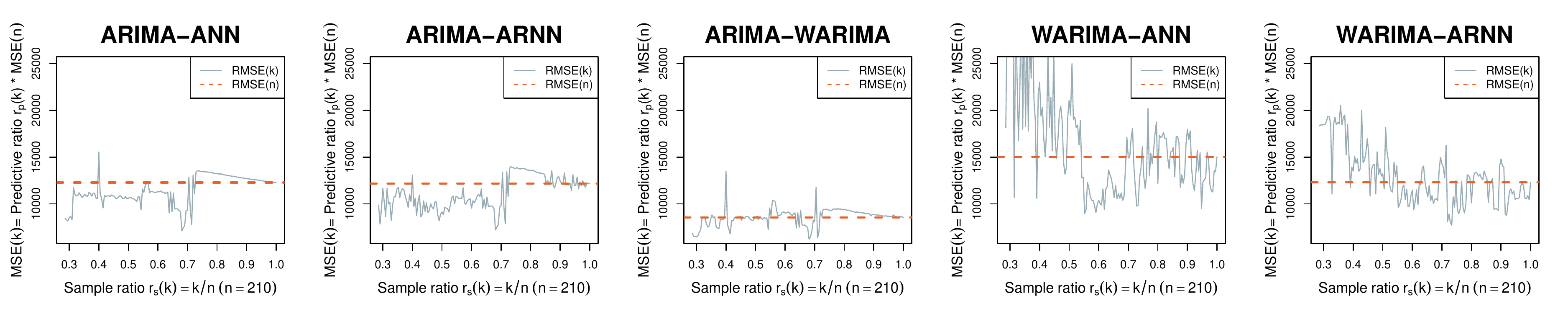}
        \caption{Hybrid}
        \label{fig:hybrid}
     \end{subfigure}
     \begin{subfigure}[b]{0.95\textwidth}
         \centering
        \centering
        \includegraphics[width=\textwidth]{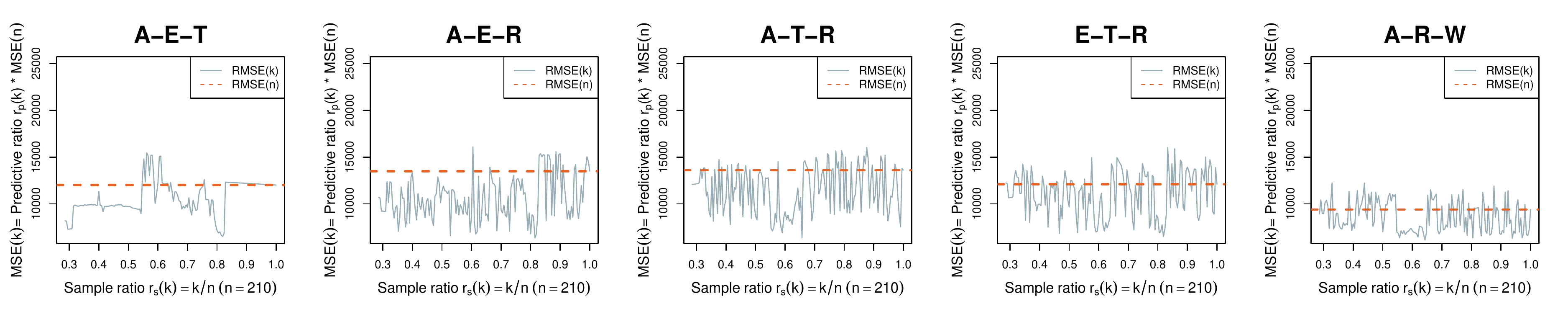}
        \caption{Ensemble}
        \label{fig:ensemble}
     \end{subfigure}
        \caption{The Pareto optimal efficiency plot for comparing twenty methods used in \cite{chakraborty2022nowcasting}. Here the y-axis is $r_p(k)$ multiplied by $MSE(n)$ to compare the performance of different methods.}  
    \label{fig:covid20}
\end{figure}

Given the results in Figure \ref{fig:covid20}, we can observe: 

(1) The red dotted line ($MSE(n)$) of almost all methods is above some part of the predictive curve using fewer PaEBack samples, which means PaEBack subsamples can yield better practical predictive performance than using complete historical data for most models. 

(2) For complex models such as ANN, which has a large number of parameters to tune, or the threshold-based method SETAR, which needs a large size of training data to decide the threshold accurately, the practical predictive performance of using less sample $k$ is not as good as utilizing complete data. However, the performance of using the total historical data is also not promising compared with other methods. For more exact comparisons, we have exhibited numerical performances in Table \ref{tab:covid}.

\begin{table}[h] 
\centering
\caption{Forecasting results for the USA test set with a 30-day horizon. We select the optimal development sample size k by minimizing the test RMSE recorded as the RMSE optim. Then we calculate the efficiency by RMSE optim / RMSE full, similarly for SMAPE.}
		\label{tab:covid}
\begin{tabular}{@{}cc | crrc | cccc@{}} \\ \toprule
   \multicolumn{2}{c}{\multirow{2}{*}{$Methods$}}&  & \multicolumn{2}{c}{$RMSE (\times 10^{-3})$} & & &
     \multicolumn{2}{c}{$SMAPE$} & \\ \cmidrule{4-5} \cmidrule{8-9} 
 & & $k^{opt}$ & {$(k^{opt})$} & {$(k=n)$} &  {$r_p(k^{opt})$} &   $k^{opt}$ & {$(k^{opt})$} & {$(k=n)$} &  {$r_p(k^{opt})$} \\ \midrule
 \multirow{3}{*}{Classical}  &
  ARIMA &  143 & 7.044 & 12.370 & 0.569 &     143 & 0.147 & 0.243 & 0.606 \\ 
  & SETAR &  205 & 8.594 & 8.594 & 1.000 &     205 & 0.173 & 0.173 & 1.000 \\ 
  & ARFIMA &  185 & \cellcolor[gray]{0.9} {6.831} & \cellcolor[gray]{0.9} \textbf{6.847} & 0.998 &     210 & \cellcolor[gray]{0.9} 0.144 & \cellcolor[gray]{0.9} \textbf{0.144} & 1.000 \\  \midrule
  \multirow{3}{*}{Smoothing} &
   ETS &  124 & 7.156 & 11.930 & 0.600 &     124 & 0.147 & 0.235 & 0.625 \\ 
  & TBATS &  160 & \cellcolor[gray]{0.9} \textbf{5.672} & 10.314 & 0.550 &     167 & \cellcolor[gray]{0.9} \textbf{0.111} & 0.207 & 0.536 \\ 
  & Theta &  114 & 10.398 & 12.234 & 0.850 &     114 & 0.196 & 0.232 & 0.841 \\ 
\midrule
  \multirow{2}{*}{Advanced} &
    WARIMA &  154 & 7.597 & 12.455 & 0.610 &     154 & 0.156 & 0.275 & 0.568 \\ 
  & BSTS &   \textbf{85} & \cellcolor[gray]{0.9} 6.821 & 7.855 & 0.868 &      \textbf{72} & \cellcolor[gray]{0.9} 0.141 & 0.158 & 0.893 \\ 
\midrule
  \multirow{2}{*}{ML} &
   ANN &  140 & 11.363 & 13.023 & 0.873 &     140 & 0.220 & 0.246 & 0.894 \\ 
  & ARNN &  178 & \cellcolor[gray]{0.9} 6.806 & 16.921 & \textbf{0.402} &    178 & \cellcolor[gray]{0.9} 0.134 & 0.330 & 0.408 \\  \midrule
  \multirow{5}{*}{Hybrid} &
    A-ANN &  143 & 7.100 & 12.282 & 0.578 &     143 & 0.148 & 0.239 & 0.619 \\ 
  & A-ARNN &  143 & 7.207 & 12.175 & 0.592 &     143 & 0.151 & 0.229 & 0.662 \\ 
  & A-WA &  143 & \cellcolor[gray]{0.9} 6.220 & 8.559 & 0.727 &     149 & \cellcolor[gray]{0.9} 0.128 & 0.170 & 0.749 \\ 
  & WA-ANN &  116 & 8.969 & 15.031 & 0.597 &     124 & 0.175 & 0.370 & 0.471 \\ 
  & WA-ARNN &  154 & 7.725 & 12.294 & 0.628 &     153 & 0.155 & 0.269 & 0.575 \\ \midrule
  \multirow{5}{*}{Ensemble} &
  A-E-T &  172 & \cellcolor[gray]{0.9} 6.516 & 12.014 & 0.542 &     172 & \cellcolor[gray]{0.9} 0.134 & 0.235 & 0.571 \\ 
  & A-E-R &  171 & \cellcolor[gray]{0.9} 6.348 & 13.493 & 0.470 &     172 & \cellcolor[gray]{0.9} 0.124 & 0.273 & \textbf{0.456} \\ 
  & A-T-R &  138 & \cellcolor[gray]{0.9}  6.359 & 13.604 & 0.467 &     138 & \cellcolor[gray]{0.9}  0.135 & 0.272 & 0.498 \\ 
  & E-T-R &  172 & \cellcolor[gray]{0.9} 6.496 & 12.101 & 0.537 &     172 & \cellcolor[gray]{0.9}  0.140 & 0.242 & 0.578 \\ 
  & A-R-W &  135 & \cellcolor[gray]{0.9} {6.137} & 9.399 & 0.653 &    135 & \cellcolor[gray]{0.9} {0.125} & 0.189 & 0.662 \\  \bottomrule
\end{tabular}
\end{table}

As Table \ref{tab:covid} shows, ARFIMA has the best predictive performance using full historical data with $RMSE (\times 10^{-3})$ around 6.847. However, after applying the PaEBack technique to all twenty models, TBATS, BSTS, the hybrid model of ARIMA-WARIMA, and all ensemble models yield better practical predictive performance while using fewer PaEBack samples (highlighted in grey). This improvement has demonstrated the effectiveness of the PaEBack technique in providing practically efficient forecasts even for non-stationary and challenging time series data such as Coronavirus.

\section{Discussions and Future Directions}
\label{sec:con}
This paper presents Pareto-Efficient Backsubsampling for Time Series data (PaEBack), which adopts a dual efficiency framework for short-term time series forecasting. While long historical time series can enhance forecasting accuracy for stationary series, this may not hold true for locally stationary or non-stationary scenarios. By employing the PaEBack framework, users can determine an appropriate subsample size of training data that achieves a balance between high forecasting accuracy and reduced data storage and processing requirements, even when the time series is assumed to be stationary.

The general framework of PaEBack is applicable universally in practice as it accommodates any appropriate forecasting model and evaluating criteria. However, for clarity, we establish the theoretical foundation for the concepts of Pareto optimal efficiency, practically predictive irrelevancy, and the asymptotically optimal PaEBack sample size based on AR time series approximation. We introduce the PaEBack-SW method for parameter selection in model training and incorporate adjusted adaptive weights to capture time series characteristics effectively. Through simulation studies, we observe the primary pattern of the efficiency curve under the oracle setting and investigate the effects of different order selection methods and model misspecification. Real-world illustrations include the comprehensive investigation of popular methods for forecasting the log return of the stock market and forecasting confirmed Coronavirus cases with challenges on highly non-stationarity.

Several extensions and theoretical developments are possible as a part of future work. Notably, the PaEBack method is not restricted to specific models or discrepancy criteria. As for future directions, the initial AR approximation can be extended to other appropriate (non-linear autoregressive) models, and the prediction criterion defined as MSE can be replaced with alternative metrics such as RMSE, MAP, SMAPE, and more without sacrificing generalization. Further research on non-linear AR processes or local stationary process models would be valuable to expand upon these findings.

Aside from the sample efficiency, we may also extend the Pareto optimality to computational complexity. Notice that if we take the Yuler-Walker estimation of an AR($p$) process as an example, the computational complexity can be approximated as $\mathcal{O}(n)$:
(i) $\mathcal{O}(p^2)$ for the $p$ by $p$ Toeplitz matrix inversion, and (ii) $\mathcal{O}(np)$ for the autocorrelation calculation, involving multiplication of the $p$ by $n$ lagged history and a data vector of length $n$. 

Consequently, if we define the computational efficiency $r_c(k) = \frac{\mathcal{O}(k)}{\mathcal{O}(n)}$ as the ratio of computational complexity using the past $k$ sample $X_{(n-k+1):k}$ to that of utilizing the complete data $X_{1:n}$, then $r_c(k) = \mathcal{O}(r_s(k))$, i.e., the concept of optimizing the computational complexity is then equivalent to optimizing the sample efficiency. Further research on this area may also be valuable.

\spacingset{.5} 
\bibliographystyle{agsm} 
\bibliography{ts.bib}  

\spacingset{1.5}

\appendixpagenumbering

\begin{appendices}

\section{Asymptotic PaEBack Predictive Ratio}	\label{appA}

In this Appendix, we provide the detailed proof of Theorem \ref{thm:1} (Asymptotic PaEBack predictive ratio) and Theorem \ref{thm:1} (Asymptotic PaEBack subsampling size) in Appendix \ref{appA1} and a straightforward illustration of the ratio computation in Appendix \ref{appA2}. A few supporting results are stated as Lemmas, and then used to establish the main result.

\subsection{Proof of Theorem \ref{thm:1}} \label{appA1}
Assume time series \{$X_t: t=1, 2, \cdots$\} follows the stationary AR(p) process with zero mean and finite variance which can be expressed as:
\begin{equation} \label{eq1:ar2def}
X_{t}={\phi}_{1} X_{t-1}+{\phi}_{2} X_{t-2} + \cdots +{\phi}_{p} X_{t-p} +  \epsilon_{t}, \quad \epsilon_{t}  \sim \text{WN}\left(0, \sigma^2\right), \forall t > p
\end{equation}
Denote the linear least predictor of $X_{n+h}$ based on the true coefficient $\phi$ as ${X_n^h}=E(X_{n+h}| X_{1:n} )$, which minimizes the mean square forecasting error.

\begin{definition} \label{def:a1h}
    Define vector $\mathbf{a}(h)=(a_1(h), a_2(h), \cdots, a_p(h))^{T}$ to facilitate the expression of $h$-step forecast, where
    \begin{equation} \label{eq:xnhorg}
    \left \{
    \begin{array}{c}
         \begin{aligned}
    {X}_{n}^{1} & = {\phi}_{1} X_{n} + {\phi}_{2}  X_{n-1} + \cdots + {\phi}_{p}   X_{n+1-p} \stackrel{}{=} a_1(1) X_n + a_2(1)X_{n-1} +\cdots ;\\
    {X}_{n}^{2} & = {\phi}_{1}   {X}_{n}^{1} + {\phi}_{2}   X_{n} + \cdots + {\phi}_{p}   X_{n+2-p} \stackrel{}{=} a_1 (2) X_n + a_2(2) X_{n-1} + \cdots ; \\
    \vdots\\
    {X}_{n}^{h} & = {\phi}_{1}   {X_{n}^{h-1}} + \cdots + {\phi}_{p}   {X_{n}^{h-p}}\stackrel{}{=} a_1(h) X_{n} + a_2(h)X_{n-1}+ \cdots.
    \end{aligned}
    \end{array}
    \right .
    \end{equation}
\end{definition}

Notice that the predictors $X_n^h$s are obtained recursively based on the past p observations $\mathcal{I}_{n-p+1}^n$ and that $a_1(h)$ has the expression in Lemma \ref{prop:a1h} with $a_1(0)=1$, $a_1(1)=\phi_1$, $a_1(2)=\phi_1^2+\phi_2$, $a_1(3)=\phi_1^3+2\phi_1\phi_2+\phi_3$, etc. Besides, $a_i(h)$ is defined as the coefficient for $X_n^h$ with regard to observation $X_{n-i+1}$. 

\begin{lemma} \label{prop:a1h}
    The iterative function $a_1(h)$ defined in Definition \ref{def:a1h} depends on the true $\phi$ through the following expression:
	\begin{align}
	a_1(h) = \mathbbm{1}(h = 0) +  \mathbbm{1} {(h \geq 1)} \left\{ 
	{\phi}_{1} a_1(h-1) + 
	\cdots +  {\phi}_{p} a_1(h-p)
	\right\}
	\label{eq:19}
	\end{align}
	where $\mathbbm{1} {(\cdot)}$ is the indicator function.
\end{lemma}

\begin{lemma} \label{prop:sigma2h}
	The MSE of the linear predictor ${X_n^h}$ is given by
	\begin{equation}
	\sigma^2_h = \sigma^2 \sum_{j=0}^{h-1}{a_1(j)}^2,
	\end{equation}
	where $a_{1(j)}$'s are given in Lemma \ref{prop:a1h}.
\end{lemma}

When the true coefficient $\phi$ is unknown, the Yule-Walker estimation based on the past $k$ observations $X_{(n-k+1):n}$ denoted as ${\phi}_{\cdot,k}$ is known to have the asymptotic multivariate normal distribution (\cite{box1970time}). Define the estimated forecast as $X_{n,k}^h$, then it is obtained by substituting the true $\phi$ in Eq. (\ref{eq:xnhorg}) with ${\phi}_{\cdot,k}$.

\begin{lemma} [asymptotic variance] \label{prop:error}
	Define the forecast error at $h^{th}$ prediction step using past $k$ observations as $e_{n}^h(k)=X_{n+h} - {\hat{X}_{n}^{h}}(k)$, the $p\times p$ matrix $M_h=\frac{\partial a(h)} {\partial \phi}$, and the vector $\vec{X}_n^{n-p+1}= (X_n, X_{n-1}, \cdots, X_{n-p+1})^T $. Then we can express
	\begin{align} \label{eq:ekt1}
	e_n^{h}(k) 
	& = -\sum_{j=0}^{h-1}a_1(j)\epsilon_{n+h-j} + (\hat{\phi}^{(k)}-\phi)^TM_h \vec{X}_n^{n-p+1},
	\end{align}
	with the asymptotic variance $v_n^{h}(k)$ given by
	\begin{align}
	v_n^{h}(k)  &= \sigma_h^2 + {w_h{(k)}} ,
	\end{align}
	where $\sigma^2_h$ is defined in Lemma \ref{prop:sigma2h}, ${w_h{(k)}}  =  k^{-1} \sigma^2 \cdot tr\{M_h^{T} \Gamma^{-1} M_h \Gamma \}$, and $\Gamma$ is the $p\times p$ covariance matrix $[\gamma(i-j)]_{i,j=1}^p$. For instance, when $h=1$ the asymptotic variance is $ \sigma_1^2 + \nu^2_1 = (1+ \frac{p}{k})\sigma^2$.
\end{lemma}
We omit the proofs of the above Lemmas as most of the calculations follow by using the results in \cite{yamamoto1976asymptotic}. Recall that the predictive performance of the model built on $X_{(n-k+1):n}$ is evaluated by $MSE(k)=\frac{1}{h}\sum_{t=1}^h{e_n^t(k)}^2$, as defined in Eq. (\ref{eq:mse}) and the relative predictive ratio is $r_p(k)=\frac{MSE(k)}{MSE(n)}$. The following Lemma then obtains the asymptotic MSE and the corresponding predictive ratio\ref{prop:as}.

\begin{lemma} [asymptotic prediction efficiency] \label{prop:as}
	Based on Lemma \ref{prop:error}, the asymptotic MSE is then obtained as
	\begin{align}
	AMSE(k) = \frac{1}{h} \sum_{j=1}^h \nu_n^j(k) & = \sum_{j=1}^h \sigma_j^2 + \frac{1}{k} \sigma^2 \sum_{j=1}^h tr\{M_j^{T} \Gamma^{-1} M_j \Gamma \},
	\end{align}
	with the corresponding asymptotic ratio as 
	\begin{align} \label{eq:asr}
	Ar_p(k) = \frac{AMSE(k)}{AMSE(n)} &= 1 + (\frac{1}{k}-\frac{1}{n}) \cdot \left( \frac{\sum_{j=1}^h \sum_{i=1}^{j-1}{a_1(j)}^2}{  \sum_{j=1}^h tr\{M_j^{T} \Gamma^{-1} M_j \Gamma \}} +\frac{1}{n}  \right)^{-1}.
	\end{align}
\end{lemma}
From Eq. (\ref{eq:asr}), it is evident that the value of $Ar_p(k)$ is positive and greater than one, decreases as $k$ increases, and converges to 1 when $r_s(k)=k_n/n\rightarrow 1$ as $n\rightarrow\infty$. With the above results, we are now ready to complete the proof of the main result.

\begin{proof}[Proof of Theorem \ref{thm:1} (i)]
Denote $A=\sum_{j=1}^h \sigma_j^2$,  $B=\sigma^2 \sum_{j=1}^h tr\{M_j^{T} \Gamma^{-1} M_j \Gamma \}$, then the asymptotic MSE can be then expressed as $AMSE(k) =  A+\frac{1}{k}B,$ with corresponding asymptotic ratio as $Ar_p(k) = \frac{AMSE(k)}{AMSE(n)} {=} \frac{A+\frac{B}{k}}{A+\frac{B}{n}} = 1 + \frac{ \frac{1}{k} - \frac{1}{n} }{ \frac{A}{B} + \frac{1}{n} }$ $\xrightarrow{n\rightarrow\infty} 1+k\frac{B}{A}$.
\end{proof}

\begin{proof}[Proof of Theorem \ref{thm:1} (ii)]
Continue with the proof of Theorem \ref{thm:1} and notice that $Ar_p(k) \leq 1+\epsilon_n$ is equivalent to 
\begin{align} \label{Eq:knopt}
\frac{ \frac{1}{k} - \frac{1}{n} }{ \frac{A}{B} + \frac{1}{n} } &< \epsilon_n  \Longleftrightarrow  k > n \cdot \left( \frac{1}{1+\epsilon_n+n\epsilon_n \frac{A}{B}} \right) \stackrel{\Delta}{=}n\cdot c_n
\end{align}
Given $\epsilon_n \rightarrow 0$ and $n\epsilon_n \rightarrow \lambda>0$ as $n \rightarrow \infty$, we could find that the $c_n$ part of Equation (\ref{Eq:knopt}) is converging to $\frac{1}{1+\lambda \frac{A}{B}}$, and hence, $k^{opt} \sim n \left(  \frac{1}{1+\lambda \frac{A}{B}}  \right)$ proved.
\end{proof}

As for the ratio $\frac{A}{B}$, it does not depend on $\sigma^2$ but only on AR coefficients $\phi_j$ for $j=1, \cdots, p$. Also, we could notice that when $h=1$, $k^{opt}\sim n(\frac{1}{1+\epsilon_n+\lambda/p})$, and thus, with larger $p$, the optimal $k$ is expected to be larger.

\subsection{Illustration of calculating the ratio A/B for AR(5)}\label{appA2}

To illustrate with the given example shown in Figure \ref{fig:simu2} where $h=3$ and the AR(5) model has the AR coefficients as $(0.5, -0.4, 0.3, -0.2, 0.1)$ with $\sigma=1$, $A=\sum_{j=1}^h \sigma_j^2=\sigma^2 \sum_{j=1}^h \sum_{i=1}^{j-1}a_1(i)^2$, and $B=\sigma^2 \sum_{j=1}^h tr\{M_j^{T} \Gamma^{-1} M_j \Gamma \}$, we could have the ratio $\frac{A}{B}=\frac{\sum_{j=1}^h \sum_{i=0}^{j-1}a_1(i)^2}{\sum_{j=1}^h tr\{M_j^{T} \Gamma^{-1} M_j \Gamma \}} \stackrel{\Delta}{=}\frac{A_h}{\sum_{j=1}^h tr_{j}}$ with $A_3=\sum_{j=1}^3 \sum_{i=0}^{j-1}a_1(i)^2=3.5225$.

When $h=1$, it is handy to derive the general $a(1)$ and $M_1$ for any AR(p) process as   $a(1)=(\phi_1, \phi_2, \cdots, \phi_p)$ and $M_1=\frac{\partial{a(1)}}{\partial{\phi}}=I_{p \times p}$. Then $tr_1=tr{M_1^T\Gamma^{-1}M_1\Gamma}= p=5$.

When $h=2$, we could derive the general format of $a(2)$ and $M_2$ for any AR(p) process as $a(2)=(\phi_1^2+\phi_2,\cdots, \phi_1\phi_i+\phi_{i+1}, \cdots, \phi_1 \phi_{p-1}+\phi_p, \phi_1\phi_p)$  and 
\begin{equation}\label{eq:m2}
M_2=\frac{\partial{a(2)}}{\partial{\phi}}=\left( \begin{array}{ccccc}
2\phi_1 &1 &0 & \cdots & 0 \\
\phi_2 & \phi_1 & 1 & \cdots & 0 \\
\vdots & \vdots & \vdots & \ddots & \vdots \\
\phi_{p-1} & 0 & 0 & \cdots & 1 \\
\phi_p & 0 & 0 & \cdots & \phi_1
\end{array}\right).
\end{equation} 
Then we could obtain $tr_2=9.9468$.

When $h=3$, the general form of $a(3)$ and $M_3$ is complicated, and we only write them for the given AR(5) process as $a(3)=(\phi_1^3+2\phi_1\phi_2+\phi_3, \phi_1^2\phi_2+ \phi_1\phi_2+\phi_1\phi_3+\phi_4, \phi_1^2\phi_3+\phi_1\phi_4+\phi_2\phi_3+\phi_5, \phi_1^2\phi_4+\phi_1\phi_5+\phi_2\phi_4, \phi_1^2\phi_5+\phi_2\phi_5)$ and
\begin{equation} \label{eq:m3}
M_3=\frac{\partial{a(3)}}{\partial{\phi}}=\left( \begin{array}{ccccc}
3\phi_1^2+2\phi_2 & 2\phi_1 &  & \cdots & 0 \\
2\phi_1\phi_2+\phi_2+\phi_3 & \phi_1^2+\phi_1 &   & \cdots & 0 \\
\vdots & \vdots & \vdots & \ddots & \vdots \\
2\phi_1\phi_5 & \phi_5 &   & \cdots & \phi_1^2+\phi_2
\end{array}\right). 
\end{equation}
We could obtain $tr_3=8.1480$.

Therefore, for our specific example, the ratio $\frac{A}{B}=\frac{3.5225}{5+9.9468+8.1480}=0.1525$.

\subsection{Illustration of estimating the ratio A/B of the stock data }\label{appA3}
When we do not know the true value of the ratio $A/B$ for real data, we can still obtain consistent estimates using the YW estimate of the AR coefficients, $\phi_i$s. For illustration we use the example featured in Section \ref{sec:stock}, where we used the stock price of Apple with $k=0.2$, $n=200$ and $h=3$. Notice that the best model turned out to be AR(2) with estimated coefficients $\hat{\phi}= (-0.2446, 0.0571)$.

Recall the ratio $\frac{A}{B}=\frac{\sum_{j=1}^h \sum_{i=0}^{j-1}a_1(i)^2}{\sum_{j=1}^h tr\{M_j^{T} \Gamma^{-1} M_j \Gamma \}} \stackrel{\Delta}{=}\frac{A_h}{\sum_{j=1}^h tr_{j}}$.
To begin with the numerator, we have $\hat{a_1}(0)=a_1(0)=1$, $\hat{a_1}(1)=\hat{\phi}_1a_1(0)=\hat{\phi}_1=-0.2446$, $\hat{a_1}(2)=\hat{\phi}_1^2+\hat{\phi}_2=0.1170$, 
and hence $A_3=\sum_{j=1}^3 \sum_{i=0}^{j-1}a_1(i)^2=3.1334$. 
When $h=1$, $\hat{tr}_1=tr_1=  p=2$. When $h=2$ and $a(2)=(\phi_1^2+\phi_2, \phi_1\phi_2)$, we have
\begin{equation}
\hat{M_2}=\left( \begin{array}{cc}
2\phi_1  & 1  \\
\phi_2 & \phi_1
\end{array} \right)_{\phi=\hat{\phi}}
\end{equation}
and $tr_2=1.4993$. 
When $h=3$ and  $a(3)=(\phi_1^3+2\phi_1\phi_2, \phi_1^2\phi_2+\phi_1\phi_2)$, we have
\begin{equation}
\hat{M_3}=\left( \begin{array}{cc}
2\phi_1^2+2\phi_2  & 2\phi_1 \\
2\phi_1\phi_2+\phi_2 & \phi_1^2+\phi_1
\end{array} \right)_{\phi=\hat{\phi}}
\end{equation}
and $tr_3=0.4819$.
Hence, the estimated ratio for this specific Apple stock data is then $\frac{\hat{A}}{\hat{B}}=\frac{3.1334}{2+1.4993+0.4819}=0.7870$. 

Given the specific sample ratio from $r = 0.1$ to $r = 0.2$, we could use Eq. \ref{eq:askopt} to obtain the values of $\lambda$ in $[5.0824, 11.4354]$, which would correspond to the efficiency loss $\epsilon_n$ in the range of only $[0.005, 0.011]$.

\section{Comparison with Other Subsampling Method} \label{sec:fkc}
\cite{fukuchi1999subsampling} has discussed subsampling methods using sliding windows with size $k$, either overlapped or non-overlapped versions, to estimate the risk of prediction for time series data. However, we have shown that without separating the validation set, as shown in Figure \ref{fig:bs} using the proposed PaEBack strategy, the evaluation of the predictive sample can embrace a full historical sample size, as indicated by Figure \ref{fig:fukuchi}.

\begin{figure}[h]
	\centering
	\begin{minipage}[t]{.95\linewidth}
		\centering
		\includegraphics[width=\textwidth]{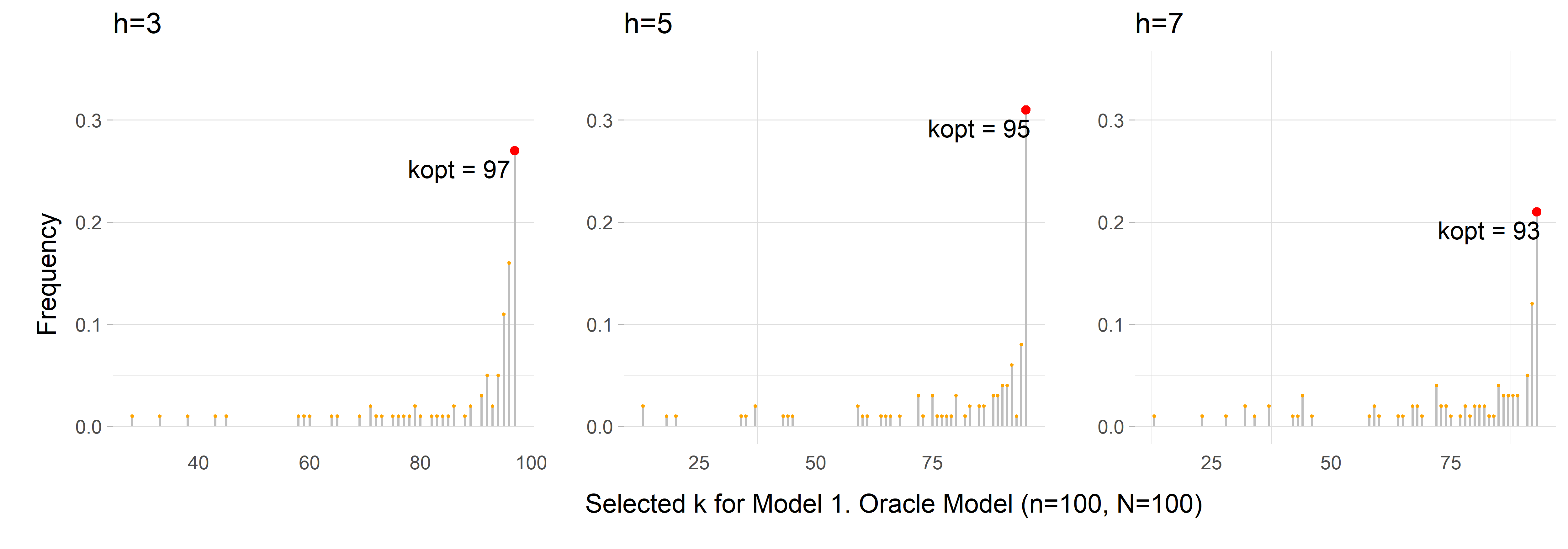}
	\end{minipage}%
	\centering
	\caption{Here the full sample size is $n+h=100$ with forecasting horizon h varying from $h=3$, $5$, to $h=7$ with 100 replicates. Simulated data is generated from AR(5) process with $\phi=c(0.5, -0.4, 0.3, -0.2, 0.1)^T$. Each training model $l$ is fitted on the development set $\mathcal{I}_{l}^{k+l-1}$ with window size equal to $k$, and evaluated on the $\mathcal{I}_{k+l}^{k+l+h-1}$. The optimal $k$ is selected to minimize the overall mean MSE for a given $k$.}
	\label{fig:fukuchi}
\end{figure}

Notice that for each window size $k$, using their strategy, there are $n-k-h+1$ models fitted, and thus, the optimal development size $k$ is selected such that
\begin{equation}
k^{opt} = \arg \min_k \frac{ \sum_{l=1}^{n-k-h+1} \frac{1}{h}\sum_{t=1}^h (X_{k+l-1+t} - \hat{X}_{k+l-1+t}^{(k)})^2}{n-k-h+1}.
\end{equation}
Moreover, we could notice that the testing sets overlap with each other, where the overall testing size $h \cdot (n-k-h+1)$ decreases as the window size $k$ increases. Thus, as shown in Figure \ref{fig:fukuchi}, the simulated results of the optimal $k$ would always embrace the largest possible window size, $k=n-h$, including the only training model and the only test set $\mathcal{I}_{n-h+1}^n$. However, this does not indicate \cite{fukuchi1999subsampling} as a not reliable subsampling method in model selection, but just not as effective as the PaEBack framework in determining the dual optimal development subsample size $k$ with more objective considerations.

\section{PaEBack with Non-Increasing Adjusted Adaptive Weight Algorithm} \label{sec:appC}
\begin{algorithm}[h]
	\caption{PaEBack Non-Increasing Adjusted Adaptive Weight Algorithm}
	\label{algo:wt}
	\begin{algorithmic}[1]
		\REQUIRE{Input vector $b$, which is the absolute value of the initial estimator $\beta^{ini}$. Set the adjusted $b^a=b$, pointer $=0$, index1 $=1$, index2 $=0$, and $|b|$ as the size of $b$.}
		\WHILE{pointer $<$ $|b|$}
		\STATE{upper $\gets$ $b_{index1}$}
		\FOR{$k\in\{2, \cdots, |b|\}$} 
		\IF{$b_k < b_{k+1}$} 
		\STATE index2 $\gets$ k; lower $\gets$ $b_{index2}$; break;
		\ENDIF
		\IF{$index2==index1$}
		\STATE $b^a$ $\gets$ rep( $b_{index1}$, $|b|$)
		\ENDIF
		\IF{index2$-$index1 $>$ 1}
		\FOR{i in 1:(index2-index1-1)}
		\STATE $b^a_{index1+i}$ $\gets$ $b^a_{index1}+ \frac{b_{index2}-b_{index1}}{index2-index1} \cdot i$
		\ENDFOR
		\ENDIF
		\ENDFOR
            \STATE{pointer $\gets$  pointer$+1$}
		\ENDWHILE
		\RETURN{Output the adjusted $b^a$ with non-negative values in non-increasing order.}
	\end{algorithmic}
\end{algorithm}

\section{Model Selection for the Log Return of Stock Prices} \label{sec:appD}

\begin{table}[H]
\centering
\caption{Bayesian information criterion (BIC) of selecting the $p$ and $q$ parameters for four models introduced in Section \ref{sec:stockm}: (v) GARCH(p,q), (vi) ARIMA(p,0,q)-GARCH(1,1), (vi) gjrGARCH(p,q), and (viii) ARIMA(p,0,q)-gjrGARCH(1,1). The models are trained on the log return of Amazon's stock price data with historical sample size $n=1000$. The smaller BIC, the better a model performs. The best values (smallest BICs) for each method (column) are bolded.}
\label{tab:garchbic}
\begin{tabular}{@{}rllrrrr@{}} 
 \\   \toprule
 & p & q & GARCH & AGARCH & gjrGARCH & AgjrGARCH \\ \midrule

(0,1) & 0 & 1 & -5.05198 & \cellcolor[gray]{0.9} \textbf{-5.26472} & -5.05198 & \cellcolor[gray]{0.9} \textbf{-5.27090} \\ 
  (0,2) & 0 & 2 & -5.04581 & -5.25908 & -5.04581 & -5.26569 \\ 
  (0,3) & 0 & 3 & -5.03919 & -5.25670 & -5.03919 & -5.26235 \\  [6pt]
  (1,0) & 1 & 0 & -5.14057 & -5.26470 & -1.35564 & -5.27084 \\ 
  (2,0) & 2 & 0 & -5.20139 & -5.25884 & 0.23821 & -5.26531 \\ 
  (3,0) & 3 & 0 & -5.21831 & -5.25611 & -5.20264 & -5.26194 \\ [6pt]
  (1,1) & 1 & 1 & \cellcolor[gray]{0.9} \textbf{-5.27141} & -5.26313 & \cellcolor[gray]{0.9} \textbf{-5.27715} & -5.26839 \\ 
  (1,2) & 1 & 2 & -5.26430 & -5.25094 & -5.27002 & -5.26200 \\ 
  (1,3) & 1 & 3 & -5.25723 & -5.25036 & -5.26290 & -5.25574 \\  [6pt]
  (2,1) & 2 & 1 & -5.26995 & -5.25699 & -5.27238 & -5.26208 \\ 
  (2,2) & 2 & 2 & -5.26304 & -5.25013 & -5.26548 & -5.25528 \\ 
  (2,3) & 2 & 3 & -5.25612 & -5.24368 & -5.25884 & -5.25074 \\  [6pt]
  (3,1) & 3 & 1 & -5.26992 & -5.25021 & -5.26663 & -5.25566 \\ 
  (3,2) & 3 & 2 & -5.26301 & -5.24364 & -5.25972 & -5.24816 \\ 
  (3,3) & 3 & 3 & -5.25610 & -5.24503 & -5.24705 & -5.24498 \\ \bottomrule
\end{tabular}
\end{table}


\section{Machine Learning Methods for COVID-19 Cases Nowcasting} \label{sec:covidsupp}

In this Section, we briefly introduce the twenty forecasting methods used in Section \ref{sec:covid} with further details available in \cite{chakraborty2022nowcasting}.

\paragraph{Classical} The classical methods include ARIMA, SETAR, and ARFIMA. (1): ARIMA is one of the most well-known linear models in time-series forecasting, typically with three parameters as ARIMA(p,d,q), where p and q stand for the order of AR and MA parts, respectively, and d represents the level of differencing to convert non-stationary data into stationary time series (\cite{box2015time}). (2): Self-exciting threshold autoregressive (SETAR) model has two parameters in SETAR (k,p) (\cite{tong1990non}). It allows parameter switching among AR(p) models with k+1 regimes. (3): Autoregressive fractionally integrated moving average (ARFIMA) model is the extension of the ARIMA by allowing non-integer values of the differencing parameter \cite{granger1980introduction}. A typical ARFIMA model has three parameters in ARFIMA(p, d, q), similar to the ARIMA model.

\paragraph{Smoothing} Smoothing methods include ETS, TBATS, and THETA. (1): Smoothing methods such as Exponential smoothing (\cite{winters1960forecasting}) are very effective in time series forecasting. The exponential smoothing state space model decomposes the time series into three-level components: the Error component (E), Trend component (T), and Seasonal component (S). (2): TBATS model uses exponential smoothing to deal with complex seasonal patterns. The name TBATS is the acronym for key features of the models: Trigonometric seasonality (T), Box-Cox Transformation (B), ARMA errors (A), Trend (T), and Seasonal (S) components (\cite{de2011forecasting}). (3): THETA method decomposes the original data into two or more theta lines and extrapolates them using forecasting models. The weighted averages of forecasts based on different theta are combined to obtain the final forecasts (\cite{assimakopoulos2000theta}).

\paragraph{Advanced} There are two advanced methods. (1): Wavelet-based ARIMA (WARIMA) transforms time series data and is most suitable for non-stationary data, unlike standard ARIMA. It applies Daubechies wavelets transformation and decomposition to the time series and removes the high-frequency components before fitting the ARIMA model to provide out-of-sample forecasts (\cite{aminghafari2007forecasting}). (2): The Bayesian structural time series (BSTS) model has been applied in \cite{scott2013predicting} to show how Google search data can be used to improve short-term forecasts of economic time series.

\paragraph{ML} Two types of ML methods have been applied. (1): Forecasting with artificial neural networks (ANN) has received increasing interest in the late 1990s and has been given special attention in epidemiological forecasting (\cite{philemon2019review}). One-layer network has been applied here to conduct time series forecasting. (2): An autoregressive neural network (ARNN(p,k)) is a modification to the simple ANN using p-lagged inputs of the time series and k number of hidden neurons in the architecture of a simple feedforward neural network (\cite{faraway1998time}).

\paragraph{Hybrid} Hybrid methods model the linear part of the series using linear methods such as ARIMA, while the residual portion is fitted using non-linear methods such as ANN or ARNN.

\paragraph{Ensemble} Ensemble methods directly use weighted averages among different forecasts.

\end{appendices}

\end{document}